%% file: main.tex
\documentclass[a4paper,USenglish,thm-restate,cleveref]{lipics-v2021}
\usepackage{graphicx} 
\usepackage{amsmath,amsthm}
\usepackage{xcolor}
\usepackage{algorithmicx}
\usepackage{algorithm} 
\usepackage{algpseudocode}
\usepackage{xifthen}
\usepackage{graphicx}
\usepackage{wrapfig}
\usepackage{dblfloatfix}
\usepackage{xspace}
\usepackage{comment}
\usepackage{url}

\AddToHook{cmd/appendix/before}{ 
    \crefalias{section}{appendix}
    \crefalias{subsection}{appendix}
}

\newcommand{\ayaz}[1]{{\color{cyan} \textbf{Ayaz: }#1}}
\newcommand{\marc}[1]{{\color{purple} \textbf{Marc: }#1}}

\algnewcommand{\BlueComment}[1]{\textcolor{blue}{\hfill\(\triangleright\) #1}}

\newcommand{\tssCombine}[2]{\mathsf{Combine}^{#1}(#2)\xspace}
\newcommand{\tssVerify}[2]{\mathsf{CombinedVerify}^{#1}(#2)\xspace}
\newcommand{\tssShareSign}[3]{\mathsf{ShareSign}_{#1}^{#2}(#3)\xspace}
\newcommand{\tssShareVerify}[3]{\mathsf{ShareVerify}_{#1}^{#2}(#3)\xspace}

\newcommand{\extValid}[1]{\mathsf{externallyValid}(#1)\xspace}

\author{Marc Dufay}{ETH Z\"{u}rich, Switzerland}{mdufay@ethz.ch}{}{}
\author{Muhammad Ayaz Dzulfikar\footnote{Work was partially done while at National University of Singapore.}}{École Polytechnique Fédérale de
Lausanne, Switzerland}{muhammad.dzulfikar@epfl.ch}{}{}
\author{Seth Gilbert}{National University of Singapore, Singapore}{seth.gilbert@nus.edu.sg}{}{}

\authorrunning{M. Dufay, M.\,A. Dzulfikar, and S. Gilbert}

\Copyright{Marc Dufay, Muhammad A. Dzulfikar, and Seth Gilbert}



\title{Predictions Can Only Help! Communication Efficient Byzantine Agreement with Predictions 
}

\date{\today}

\nolinenumbers

\begin{document}
\thispagestyle{empty}

\begin{titlepage}
\thispagestyle{empty}

\maketitle
\input{sections/abstract}

\end{titlepage}

\input{sections/introduction}
\input{sections/related_work}
\input{sections/preliminaries}
\input{sections/technical_overview}

\input{sections/using_prediction}
\input{sections/unauthenticated}
\input{sections/authenticated}
\input{sections/conclusion}

\bibliographystyle{plainurl}
\bibliography{references}

\appendix
\input{appendix/missing_proof_using_class_pred}
\input{appendix/strong_to_external}

\end{document}

%% file: sections/abstract.tex
\begin{abstract}
    In Byzantine agreement with predictions each process begins with an input value and some (unreliable) prediction bits. Recently, it has been shown that with \emph{classification predictions}---where the predictions predict each process to be honest or faulty---Byzantine agreement can be completed more quickly than without predictions, circumventing the traditional $\Omega(f)$ round lower bound. However, existing algorithms either handle limited prediction errors or send too many messages. Moreover, they all exchange $\Omega(n^3)$ bits---enough to allow the processes to approximately agree on the classifications. In fact, it almost seemed necessary to share a significant number of prediction bits if one wanted to tolerate a high number of incorrect predictions.

    In this paper, we show that this high level of communication is not inherent to a round-efficient protocol with predictions. We provide an unauthenticated algorithm with near-optimal $\tilde{\mathcal{O}}(n^2)$ communication complexity and optimal resilience $t < n/3$. Furthermore, with authentication, we give an algorithm with optimal $\mathcal{O}(n^2\kappa)$ communication complexity (where $\kappa$ is a security parameter) and near-optimal resilience $t < (\frac{1}{2} - \epsilon)n$ for any constant $\epsilon > 0$. All of our results have optimal round complexity for any number of errors in the predictions. 
    
\end{abstract}

%% file: sections/introduction.tex
\section{Introduction}

\subparagraph{Byzantine agreement}
Byzantine agreement is a central problem in distributed computing. It appears at the heart of many distributed protocols, like state machine replication~\cite{ Abd-El-MalekGGRW05, castro1999pbft}, blockchain protocols~\cite{buchman2016tendermint, crain2018dbft, gilad2017algorand}, and many others. In Byzantine agreement, there are $n$ processes that want to agree on a decision. Here, up to $t$ processes may be Byzantine, i.e., deviate arbitrarily. 
It is known that, in general, at most $t < n/3$ faults can be tolerated; and in the synchronous, authenticated setting, at most $t < n/2$ faults can be tolerated. 
Despite its importance, Byzantine agreement is inherently expensive. It is well-known that (synchronous) deterministic Byzantine agreement requires $\Omega(f)$ rounds \cite{dolev1990earlystopping} and exchanges $\Omega(n + f^2)$ bits \cite{dolev1985bounds}, where $f \le t$ is the actual number of faults during the execution. (Here, $t$ is a parameter of the algorithm, while $f$ is unknown and set by the adversary.)

\subparagraph{Predictions.}
To cope with malicious attacks, real-world distributed applications often employ network monitoring tools that can flag suspicious processes. These tools might be based, for example, on AI, such as Darktrace~\cite{darktrace} and VectraAI~\cite{vectraai}. Intuitively, if we can correctly detect faulty processes, then we should be able to execute protocols more efficiently, leading to faster agreement (or faster block production on a blockchain, etc.). 

Alas, these tools typically do not provide perfect detection of malicious users: sometimes an honest user is incorrectly flagged as malicious, and sometimes a malicious user goes undetected.  Therefore, one might imagine designing an algorithm with the following property: (i) when the detection mechanism is accurate, the algorithm will perform better, and (ii) when it is not, the algorithm will do no worse than a protocol with no such detection mechanism.

\subparagraph{Faster agreement with predictions}
This problem was explored by Ben-David, Dzulfikar, Ellen, and Gilbert~\cite{bendavid2025predictions}; following the paradigm of \emph{algorithms with predictions}, they study Byzantine agreement with \emph{classification predictions}. In this problem, along with their inputs, each process also receives an $n$-bit string denoting an (unreliable) prediction indicating whether each other process is honest or Byzantine. A key parameter for performance is the number of incorrect bits in the predictions, $B \le n^2$. They show that predictions do not help in reducing the communication complexity (number of bits exchanged), i.e., the $\Omega(n + f^2)$ lower bound still holds. However, classification predictions do help in reducing the round complexity to $\mathcal{O}(\min\{B/n, f\})$ rounds. For example, when the predictions are mostly correct, e.g., $B \in \mathcal{O}(1)$, then processes can decide in $\mathcal{O}(1)$ rounds (which is a significant improvement on the $f+1$ rounds otherwise required). Furthermore, when the predictions are quite wrong, the algorithm still performs as well as existing algorithms. Moreover, they also showed that this is asymptotically tight by giving a matching $\Omega(\min\{B/n, f\})$ lower bound.

However, these new algorithms came with several caveats. First, in the unauthenticated setting, their algorithms can only handle a limited number of prediction errors. Moreover, all of their algorithms exchange $\Omega(n^3)$ bits.  This was because all the proposed protocols required all processes to share their predictions, and this was needed to approximately agree on the classification predictions. Alas, this global reconciliation is expensive!  Hence, it is natural to ask: \emph{can we design algorithms that achieve optimal $\mathcal{O}(\min\{B/n, f\})$ round complexity for any $B$, and do so with optimal communication complexity?}

\subparagraph{Our contributions} 
In this paper, we show that cubic communication complexity is not inherent, but simply a consequence of the global reconciliation used in earlier algorithms.  We present several communication-efficient algorithms, unauthenticated and authenticated, that decide in $\mathcal{O}(\min\{B/n, f\})$ rounds for any number of prediction errors $B$. First, in the unauthenticated setting, we achieve optimal round complexity for any number of prediction errors, while exchanging a near-optimal $\tilde{\mathcal{O}}(n^2)$ bits. Second, taking a different approach with the help of cryptography, we achieve $\mathcal{O}(n^2\kappa)$ communication complexity (where $\kappa$ is a security parameter).


Our first result is built upon the \emph{guess-and-double} structure from \cite{bendavid2025predictions}. In a nutshell, by assigning processes into \emph{groups} and doing \emph{leader election on groups} to choose a committee, we build a Byzantine agreement protocol whose round complexity scales with the prediction errors. To break the cubic barrier, we show that group leader election can be achieved by sending only $\mathcal{O}(\log n)$ bits to parties of the given group. Even though some groups may fail to elect an honest leader, we prove that only a small proportion of the groups may fail, which is sufficient. This yields a resilience of $t < n/3$, and by plugging it into the aforementioned structure, we obtain the following:
\begin{restatable}{theorem}{unauthsubcubic}
    \label{thm:unauth_subcubic_bits}
    There is an unauthenticated algorithm for Byzantine agreement that tolerates up to $t < n/3$ faulty processes. Given classification predictions with $B$ error bits, the algorithm terminates in $\mathcal{O}(\min\{B/n, f\})$ rounds. Furthermore, the algorithm exchanges $\mathcal{O}(n^2 \log (\min\{B/n, f\}))$ messages and $\mathcal{O}(n^2 \log (n) \log (\min\{B/n, f\}))$ bits.
\end{restatable}

Finally, we give a more efficient authenticated algorithm. Cryptography significantly simplifies the leader election process, allowing processes to attach succinct proofs of leader legitimacy, which enables election with a higher level of resilience.  This results in near-optimal communication and resilience simultaneously, with $t < (\frac{1}{2}-\epsilon)n$:
%

\begin{restatable}{theorem}{authquad}
    \label{thm:auth_quad_bits}
    There is an authenticated algorithm for Byzantine agreement that tolerates up to $t < (\frac{1}{2}-\epsilon)n$ faulty processes for any constant $0 < \epsilon < \frac{1}{2}$. Given classification predictions with $B$ error bits, the algorithm terminates in $\mathcal{O}(\min\{B/n, f\})$ rounds. Furthermore, the algorithm exchanges $\mathcal{O}(n^2)$ messages and $\mathcal{O}(n^2\kappa)$ bits.
\end{restatable}

A summary of our work, compared with existing results, can be seen at \Cref{table:complexities_summary}.

\begin{table}[h]
\centering
\begin{tabular}{ |p{1.8cm}|p{1.6cm}|p{2cm}|p{1.5cm}|p{1.5cm}|p{3cm}|}
 \hline\hline
 Algorithm & Resilience & Max. Errors & Message & Bit & Cryptography\\
 \hline\hline
 \\[-1em]
 \cite{bendavid2025predictions} & $n/3$ & $\mathcal{O}(n^{1.5})$ & $\tilde{\mathcal{O}}(n^2)$ & $\mathcal{O}(n^3)$ & None\\[1pt]
 \hline
 \\[-1em]
 \textbf{This - 1} & $n/3$ & $\mathcal{O}(n^{2})$ & $\tilde{\mathcal{O}}(n^2)$ & $\tilde{\mathcal{O}}(n^2)$ & None \\[1pt]
 \hline
 \\[-1em]
 \cite{bendavid2025predictions} & $(\frac{1}{2}-\epsilon)n$ & $\mathcal{O}(n^{2})$ & $\tilde{\mathcal{O}}(n^3)$ & $\tilde{\mathcal{O}}(n^4\kappa)$ & Threshold Signature\\[1pt]
 \hline
 \\[-1em]
 \textbf{This - 2} & $(\frac{1}{2}-\epsilon)n$ & $\mathcal{O}(n^{2})$ & $\mathcal{O}(n^2)$ & $\mathcal{O}(n^2\kappa)$ & Threshold Signature \\[1pt]
 \hline\hline
 \end{tabular}
    \caption[Comparison of Byzantine Agreement with Classification Predictions.]{Comparison of Byzantine agreement with classification predictions. All algorithms terminate in the optimal $\mathcal{O}(\min\{B/n, f\})$ rounds when the number of prediction errors $B$ is within their respective maximum errors. $\epsilon \in \Omega(1)$ is a small positive constant. $\kappa$ denotes a security parameter. We use $\tilde{\mathcal{O}}(\cdot)$ to hide $\text{polylog}(n)$ factors. While the authenticated algorithm from \cite{bendavid2025predictions} only uses PKI, here we present the bit complexity when it is implemented with threshold signatures for fairer comparison.}
\label{table:complexities_summary}
\end{table}

\subparagraph{Paper Organization}
We discuss further related work at \Cref{section:related_work}. Next, we 
describe the system model and preliminaries at \Cref{section:preliminaries}. We present the high-level overview of our techniques at \Cref{section:technical_overview}. We then discuss how the classification predictions will be used at \Cref{section:using_predictions}. \Cref{section:unauthenticated} is dedicated for our unauthenticated algorithms, and \Cref{section:authenticated} is dedicated for our authenticated algorithm. Finally, we conclude at \Cref{section:conclusion}.

%% file: sections/related_work.tex
\section{Related Work}\label{section:related_work}

Since the seminal work by Lamport, Pease, and Shostak \cite{LamportSP82, PeaseSL80}, Byzantine agreement has been widely studied. It is known that any deterministic protocol needs $\Omega(t)$ rounds \cite{dolev1983authenticated, FischerL82} and $\Omega(n + t^2)$ messages \cite{dolev1985bounds}.

One way to circumvent these lower bounds is by focusing on the actual number of failures during the execution, $f$. In particular, we can design \emph{early-stopping} protocols that terminate in $\mathcal{O}(f)$ rounds \cite{lenzen2022phase}, which is also asymptotically tight \cite{dolev1990earlystopping}.

Another promising way to circumvent existing limitations is by using (unreliable) predictions \cite{mitzenmacher2020algorithmspredictions, mitzenmacher22predictionscomm}. The goal is to perform better when the predictions are accurate (\emph{consistent}) and still perform as well as algorithms without predictions when the predictions are inaccurate (\emph{robust}). In the sequential setting, this resulted in several exciting results, including ski rental~\cite{purohit18onlineviaml}, load balancing~\cite{0001X21}, P2P network \cite{dallot2025laslinlearningaugmentedpeertopeernetwork}, and online graph algorithms~\cite{AzarPT22, choo24onlinebip}. In the distributed setting, there have been several interesting results for graph algorithms \cite{boyar2025distributedgraphalgorithmspredictions}, contention resolution \cite{gilbert2021contentionpred}, and, most relevant to our work, Byzantine agreement \cite{bendavid2025predictions, dallot2026resilientbyzantineagreementpredictions}.\footnote{See \url{https://algorithms-with-predictions.github.io/} for recent works in this topic.}

In their work \cite{bendavid2025predictions}, Ben-David, Dzulfikar, Ellen, and Gilbert study Byzantine agreement where, apart from its input, each process also receives (potentially different) prediction bits. Surprisingly, they show that \emph{any} form of predictions cannot reduce the communication cost -- Byzantine agreement still exchanges $\Omega(n + f^2)$ messages. They then focus on a form of predictions they call \emph{classification predictions}, where each process receives an $n$-bit string predicting whether a process is honest or faulty. They show that with $B$ wrong prediction bits, Byzantine agreement can be done in $\mathcal{O}(\min\{B/n, f\})$ rounds, which is much faster when the predictions are quite accurate. Moreover, they also show that this is asymptotically tight. They gave two algorithms with trade-offs: an unauthenticated algorithm that can tolerate $B \le n^{1.5}$ wrong prediction bits (and otherwise, terminates in $\mathcal{O}(f)$ rounds), and an authenticated algorithm that can tolerate any number of wrong prediction bits, but sends more messages. Furthermore, all of their algorithms exchange $\Omega(n^3)$ bits to approximately agree on the classification of each process.
In this work, we show that those limitations can be avoided.

An area that is related to our work is \emph{accountability}. In some sense, the classification predictions can be seen as unreliable accountability. There have been exciting works on implementing accountable Byzantine agreement~\cite{civit2025scalable,civit21polygraph,civit22abc,CivitGGGKMS22,ShengWNKV21} -- that is, on detecting malicious processes. Some works also study how accountability can be used to do Byzantine agreement and broadcast more efficiently \cite{civit2024dare, civit2025repeated, wan2023amortizedbb}. They leverage different notions of accountability to achieve communication-optimal Byzantine agreement, in the single-shot and multi-shot setting.

Another line of work that is relevant is \emph{failure detectors}. Initiated by Chandra and Toueg~\cite{chandra96unreliabledetector}, the failure detector is a module that is eventually perfect, updating its detections over time until it correctly identifies every faulty process. While initially stated for crash failures, the work by Malkhi and Reiter later addresses Byzantine failures as well~\cite{malkhi97intrusiondetection}. These works focus on implementing the failure detectors and finding the ``weakest'' failure detectors needed to solve Byzantine agreement~\cite{chandra1996weakest}. In contrast, our work focuses on using classification predictions (unreliable detection) to do agreement more efficiently.

%% file: sections/preliminaries.tex
\section{Preliminaries}\label{section:preliminaries}

\subparagraph{Processes}
We consider a system of $n$ processes $\Pi = \{p_1, p_2, \dots, p_n\}$. Up to $t$ processes may deviate from the prescribed protocol. We denote the processes that follow the protocol as \emph{honest} or \emph{correct}, and those who do not as \emph{Byzantine} or \emph{faulty}. We denote $f$ with $0 \le f \le t$ as the actual number of Byzantine processes in an execution.

\subparagraph{Network}
We consider a \emph{synchronous} network, where an algorithm is executed in a round-to-round manner: in each round, a process sends its messages, receives messages, and updates its local state. Furthermore, each process can communicate directly with any other process.

\subparagraph{Cryptography}
For our authenticated protocols, we use a \emph{public-key infrastructure} (PKI). Each process has a pair of a public key and a private key. By default, each message will be signed by its sender. We also use a $(k, n)$-\emph{threshold signature scheme} \cite{Libert2016, Shoup00}, where we typically use $(t+1, n)$, $(n-t, n)$, and $(\lceil (n+1)/2 \rceil, n)$ as the parameters. In a $(k, n)$-threshold signature scheme, each process holds a distinct private key; there exists a single public key. Using its private key, $p_i$ produces a partial signature for a message $m$ via $\tssShareSign{i}{k}{m}$. This partial signature can be verified via $\tssShareVerify{i}{k}{m, \mathit{psig}}$. Next, given a set $S$ of $k$ partial signatures for the same message, a process can produce a threshold signature for that message via $\tssCombine{k}{\{\mathit{psig}_i\}_{p_i \in S, |S| = k}}$. Finally, a process can verify such a threshold signature for a message $m$ by invoking $\tssVerify{k}{m, \mathit{tsig}}$. Where appropriate, we omit explicit invocations of  $\tssShareVerify{\cdot}{k}{\cdot}$ and $\tssVerify{k}{\cdot}$. Importantly, under a security parameter $\kappa$, each signature under this scheme consists of $\mathcal{O}(\kappa)$ bits.\footnote{We assume $\kappa > \log(n)$ to avoid an adversary with exponential computational power and collision between cryptographic signatures.}

\subparagraph{Predictions}
Each process $p_i$ also receives as a prediction a string $a_i$. We underline that each process may receive a different string. Moreover, we do not make any assumption about the prediction strings. In this paper, we consider \emph{classification predictions}, where each $a_i$ is a binary string of length $n$. We say $p_i$ predicts $p_j$ as honest if $a_i[j] = 1$, and Byzantine otherwise. A prediction bit is correct if it matches reality, and wrong otherwise. We denote $B$ as the number of wrong bits in the predictions, which is the number of pairs $(i, j)$, where $p_i$ is honest and either (1) $a_i[j] = 0$ and $p_j$ is honest, or (2) $a_i[j] = 1$ and $p_j$ is faulty. (Note that the wrong bits in the predictions of faulty processes are not counted.)

\subparagraph{Byzantine Agreement}
The problem of Byzantine agreement is defined as follows. Each honest process $p_i$ proposes its input $v_i$, and outputs a decision $d_i$ such that the followings are satisfied:
\begin{itemize}
    \item \emph{Agreement}: for each honest process $p_i$ and $p_j$, $d_i = d_j$.
    \item \emph{Strong unanimity}: if each honest process proposes the same value $v$, then only $v$ can be decided.
    \item \emph{Termination}: each honest process eventually decides.
\end{itemize}

The \emph{round complexity} of a protocol is defined as the number of rounds until all honest processes decide. Then, the \emph{message complexity} of a protocol is defined as the number of messages that honest processes send until all honest processes decide. Similarly, the \emph{communication complexity} or \emph{bit complexity} of a protocol is defined as the number of bits that honest processes send throughout the protocol. In this paper, we assume the size of each input to be constant.

%% file: sections/technical_overview.tex
\section{Technical Overview}\label{section:technical_overview}

\subparagraph{Using the Classification Predictions} 
We first revisit how the predictions are used in \cite{bendavid2025predictions}, specifically their unauthenticated result. First, they use a voting procedure, where each process broadcasts its prediction string. Then, a process $p_i$ classifies $p_j$ as honest if it receives more than $n/2$ strings that predict $p_j$ as honest; otherwise, it classifies $p_j$ as Byzantine. They observe that if there are $f < (\frac{1}{2}-\epsilon)n$ faults, then there are at most $\mathcal{O}(B/n)$ processes that are misclassified by some honest process, where $B$ is the number of erroneous prediction bits. Next, suppose each process orders processes based on their classification, putting those classified as honest first. Each process may have a different ordering. However, they showed that, using those orderings in phase-king style \cite{lenzen2022phase}, we can implement Byzantine agreement that terminates in $\mathcal{O}(k)$ rounds with at most $k$ misclassified processes when $k \in \mathcal{O}(\sqrt{n})$.
They then combine that agreement protocol alongside an \emph{early-stopping Byzantine agreement} that terminates in $\mathcal{O}(f)$ rounds, glued together by a weaker agreement primitive called \emph{graded consensus}. By guessing the actual number of $k$ and $f$ and limiting the rounds executed in each protocol, along with doubling the estimate when decision is not reached yet, they achieve Byzantine agreement that terminates in $\mathcal{O}(\min\{B/n, f\})$ rounds as long as $B \in \mathcal{O}(n^{1.5})$.

\subparagraph{Group-based Approach}
An interesting alternative approach for achieving better tolerance to prediction errors is partitioning the processes into groups and choosing a representative leader from each group to help drive the agreement protocol.

One option, suggested in the conclusion of~\cite{bendavid2025predictions},
is to deterministically divide processes into $3k$ similar-sized groups.  Each process picks a process in the $i$-th group as its $i$-th leader. They hinted that there will be a ``good group'' where all honest processes choose the same honest process as their leader, which is a sufficient condition for success.  


In this paper, we explore such group-based approaches, using different grouping and election procedures to achieve better communication complexity while maintaining good round complexity.  We present several implementations to pick a leader from a group, which we call \emph{group leader election}. The election is guaranteed to elect the same honest leader for all processes if the group is ``good''. For example, we say a group is good if (1) it has at least one honest process, and (2) no misclassified process. 
%

We can observe, for example, if $t < n/2$, at most $k$ processes are misclassified, and we have $m = 3k$ groups, then at least one of the groups will necessarily be good.\footnote{Specifically, there are at most $k$ groups with misclassified processes, and there are at most $\frac{t}{n/m} < \frac{m}{2} = \frac{3k}{2}$ groups without honest process.  So there are at most $5k/2 < 3k$ groups that are not good.} We use this approach in the authenticated setting to improve the communication complexity of our protocol.

\subparagraph{Optimal Resilience Grouping via Extractor}
In the unauthenticated setting, the definition of a good group given above is too weak. It requires all-to-all communication to perform leader election over such a group. As a consequence, protocols based on this approach will end up with a $\Omega(n^3)$ bit complexity. Our key idea is to first strengthen the ``good group'' definition into what we call \emph{resilient} groups. Specifically, a group is resilient only if it has an honest majority (rather than just having at least one honest process) and no misclassified process. This allows us to use a committee-based approach: all processes only need to send (some part of) their predictions to those in the committee, who do the election among themselves, and then each process elects the process chosen by the majority of the committee.
However, this stronger definition comes at the cost of a weaker resilience if we want to ensure a high proportion of resilient groups. Specifically, with some simple approach to partition processes into groups, having less than a third of non-resilient groups would require a suboptimal $t < n/6$ resilience. 

To overcome this issue, we rely on an \emph{extractor}-based approach to determine our committees. Here, the extractor we use is a bipartite multigraph that connects processes (left set) to their groups (right set), with several useful properties. 
The use of an extractor to build the groups ensures that misclassified processes are (almost) evenly spread among all committees, and therefore, most groups will end up resilient. 
This approach is more complex as groups are not disjoint anymore; a process may be part of $\mathcal{O}(1)$ groups. However, it is better in the sense that we can still make the proportion of non-resilient groups arbitrarily small while preserving the optimal $t < n/3$ resilience.

\subparagraph{Near-Optimal Communication and Optimal Resilience without Authentication}
For our first result, we implement Byzantine agreement with predictions using the extractor-based grouping, where the leaders from the groups form a committee to execute agreement.  Each process then decides on the majority value output by the committee members. 

Here, we set the number of groups $m$ to be $\Theta(\hat{k})$ such that when at most $\hat{k}$ processes are misclassified\footnote{Looking ahead, we later call such processes in the unauthenticated setting as \emph{weakly misclassified processes}.}, there will be at least $\frac{5m}{6}$ resilient groups. (In the full protocol, we use $\hat{k}$ to denote the current estimate of misclassified processes during the guess-and-double strategy.) 

We can do the group leader election using some voting procedure. The easiest approach is to just vote by broadcasting the prediction string regarding the group. However, this would send $\Omega(n^3)$ bits when the group size is $\Omega(n)$, which is too much. Instead, we exploit the structure of the resilient group. Initially, each process votes for the smallest $\mathcal{O}(1)$ processes in the group it predicts as honest. Then, each member of the group picks the smallest process in the group having more than $n/2$ votes as the leader. Finally, all processes elect the process elected by the majority of the group. Importantly, as each process only votes on a small number of processes, the election may fail even if the group is resilient. However, if the election fails in a resilient group, it means there are $\Omega(n)$ prediction errors corresponding to the processes in that group. As we later show, combined with the parameters we use for the grouping, this implies that less than $\frac{m}{6}$ resilient groups may fail the election.

Although each process may perceive a different committee, they will contain a common subset of more than $\frac{5m}{6}-\frac{m}{6} = \frac{2m}{3}$ (unknown) honest processes. This honest core of committee members (even if unknown) ensures that the agreement will succeed.  Thus, when there are at most $\hat{k}$ misclassified processes, the protocol correctly solves Byzantine agreement. Moreover, the protocol terminates in $O(\hat{k})$ rounds and exchanges $O(n^2)$ bits. When plugged into the framework from \cite{bendavid2025predictions}, we obtain a Byzantine agreement protocol tolerating up to $t < n/3$ faults that terminates in $O(\min\{B/n, f\})$ rounds while exchanging $\tilde{\mathcal{O}}(n^2)$ messages and $\mathcal{\tilde{O}}(n^2)$ bits.

\subparagraph{Optimal Communication and Near-Optimal Resilience with Authentication}

Our authenticated algorithm, apart from having optimal $\mathcal{O}(n^2\kappa)$ communication, is arguably simpler. First, we do not require the use of the stronger groups defined above; we only need the existence of an honest process in a good group. Moreover, we can rely on the standard reduction from strong unanimity to Byzantine agreement with \emph{external validity} (also known as \emph{validated Byzantine agreement}), with $\mathcal{O}(n^2 \kappa)$ bits and $\mathcal{O}(1)$ rounds \cite{ civit2024dolevtight, civit2024dare}. Afterwards, we implement the validated Byzantine agreement with a leader-based protocol where the leaders are chosen alternating between round-robin and using group leader election results, where the length before switching follows the guess-and-double strategy from before. Finally, our algorithm builds on an efficient group leader election and validated Byzantine agreement with the help of cryptography.

In our group leader election, each process sends its signature to every process it predicts honest. Then, every process $p_i$ that receives more than $n/2$ such signatures can broadcast them (or a smaller aggregate) to tell other processes to classify $p_i$ as honest. By having each process elect the smallest process in the group that it classifies as honest, all honest processes will elect the same honest process in a good group. Furthermore, processes only need to send signatures and convince others that they are honest once; after that, the group leader election can be done without communication. Using \emph{threshold signatures} to aggregate the signatures, the group leader elections can be done by spending $\mathcal{O}(n^2\kappa)$ bits and $\mathcal{O}(1)$ rounds at the start.

Lastly, our validated Byzantine agreement implementation is heavily based on some aspects of the Byzantine broadcast protocol from \cite{wan2023amortizedbb}. However, we need to take into account that each process may have different leaders, which will break the safety of the original protocol. Roughly speaking, in the original protocol, processes will forward the message from the leader to other processes. They must accept the forwarded message if it is signed by the leader, which is the same process for everyone. In a nutshell, the problem with having different leaders is that the process does not know whether it should accept the message (because the forwarder is honest) or not.

The fix is quite simple. At every phase, each process sends a signature to the process it considers as its leader. Next, a leader must attach $\lceil \frac{n+1}{2} \rceil$ such signatures to its messages. Thus, when receiving a forwarded message, a process must accept it if it contains $\lceil \frac{n+1}{2} \rceil$ signatures for the original sender. This fixes the issue and also allows us to limit the number of `leaders' to be constant during every phase. Using threshold signatures, the communication complexity of the validated Byzantine agreement will be $\mathcal{O}(n^2\kappa)$, and will terminate as soon as all honest processes have the same honest leader. 

Altogether, our authenticated protocol has near-optimal resilience with $t < (\frac{1}{2}-\epsilon)$, deciding in $\mathcal{O}(\min\{B/n, f\})$ rounds by exchanging $\mathcal{O}(n^2\kappa)$ bits.

%% file: sections/using_prediction.tex
\section{Using Classification Predictions}\label{section:using_predictions}

In this section, we explain how we are using classification predictions, define several notions of ``good'' groups, and state key combinatorial lemmas on the existence of those groups.

In our algorithms, we use classification predictions to enable each process to pick a sequence of processes. This sequence will be used, for example, as a sequence of leaders in a leader-based protocol. This sequence is obtained as follows. First, for a parameter $m$, we will partition the processes into $m$ groups. In each group, processes participate in what we call the \emph{group leader election} to elect the leader of the group. Then, the elected $m$ leaders will form the sequence of processes. Importantly, the result of the group leader election will depend on the quality of predictions regarding the processes in that group. We will show that, when a group is ``good'', then all honest processes will likely select the same honest process as the leader.  Furthermore, we will also show that when $m$ is chosen properly, we can obtain (many) ``good'' groups.  Thus, there will be (many) indices in the sequence where everyone has the same honest leader.  

Our first goal is to define the group leader election and what it means for a group to be good. To do that, we first define misclassified process. For a given process $p$, we say that it has $b$ prediction errors if $p$ is honest and exactly $b$ honest processes predict $p$ as Byzantine, or $p$ is Byzantine and exactly $b$ honest processes predict $p$ as honest. We remark that the sum of the prediction errors among all processes is exactly $B$. We are now ready to define misclassified process.

\begin{definition}[Misclassified Process]
    A process $p_i$ is misclassified if $f + b \geq n/2$ where $b$ is the number of prediction errors for $p_i$. \footnote{Let us remark that our definition is stronger compared to the one from \cite{bendavid2025predictions}, in the sense that (1) a misclassified process according to \cite{bendavid2025predictions} is also misclassified according to our definition, but (2) our misclassified process might not be misclassified according to \cite{bendavid2025predictions}.}
\end{definition}

The motivation behind this definition is the following: One simple way to improve the quality of the predictions is to exchange the predictions and use the majority of the prediction bits to classify a process as either Byzantine or honest. That is, when more than $n/2$ processes predict $p_i$ as honest, then $p_j$ classifies $p_i$ as honest and otherwise, Byzantine. Hence, a process $p_i$ can be misclassified if, at the start, sufficiently many honest processes predict $p_i$ wrongly. Moreover, using a similar observation as the one from \cite{bendavid2025predictions}, we can obtain the following.
\begin{observation}
    If $f < (\frac{1}{2} - \epsilon)n$ for some constant $0 < \epsilon < \frac{1}{2}$, then there are at most $\mathcal{O}(B/n)$ misclassified processes.
\end{observation}

This is because at least $\epsilon n \in \Omega(n)$ incorrect prediction bits from honest processes are needed to misclassify a process.

In the unauthenticated setting, considering misclassified processes is not enough to obtain an efficient algorithm. Instead, we consider \textit{weakly misclassified} processes.

\begin{definition}[Weakly Misclassified Process]
    A process $p_i$ is weakly misclassified if $f + b + n/10 \geq n/2$ where $b$ is the number of prediction errors for $p_i$.
\end{definition}

We remark that if a process is misclassified, it is also weakly misclassified. The main idea is that the $n/10$ ``slack'' allows us to get stronger properties for our protocols, which is an important ingredient for obtaining efficient communication. We only use the concept of weakly misclassified processes for our unauthenticated protocol, where $t < n/3$. Similar to the misclassified process, we can observe the following.
\begin{observation}
    If $f < n/3$, then there are at most $\mathcal{O}(B/n)$ weakly misclassified processes.
\end{observation}

Next, we explain what it means for a group to be ``good''. We consider two definitions, based on the strength of the misclassification.


\begin{definition}[Good Group]
    \label{definition:good_group}
    A group $G$ of $|G|$ processes is good if none of its processes is misclassified and at least one of its process is honest.
\end{definition}

\begin{definition}[Resilient Group]
    \label{definition:resilient_group}
    A group $G$ of $|G|$ processes is resilient if none of its processes is weakly misclassified and the majority of its processes are honest.
\end{definition}

Note that while a good group only needs at least one of its process to be honest, resilient group needs a majority of its members to be honest. Furthermore, good group will be useful for our authenticated algorithm, while resilient group will be useful for our unauthenticated algorithm. As we will see, if it possible to efficiently perform leader election for good group in the authenticated setting and for a resilient group in an unauthenticated setting.

We are now ready to define the \emph{group leader election} problem. 

\begin{definition}[Group Leader Election]
    Each process $p_i$ is input the same group $G$ along with its prediction bits $a_i$. Each process outputs $\mathit{\ell_i}$ such that the following are satisfied:
    \begin{itemize}
        \item Agreement: for each honest process $p_i$ and $p_j$, $\mathit{\ell}_i = \mathit{\ell}_j$.
        \item Validity: $\mathit{\ell}_i \in \mathit{G}$ and $\mathit{\ell}_i$ is an honest process.
    \end{itemize}
\end{definition}

While leader election is typically not useful for deterministic protocols (because the adversary can corrupt the elected leader), this is not the case when predictions are available. In \Cref{algorithm:simple-group-leader}, we show a simple implementation that solves this problem for good groups. In a nutshell, all processes exchange all predictions they have regarding the processes in the group. Then, each process outputs the smallest process that more than $n/2$ processes predict as honest.

\begin{algorithm} [h]
\caption{Simple Group Leader Election: Pseudocode (for process $p_i$)}
\label{algorithm:simple-group-leader}
\begin{algorithmic} [1]
\footnotesize

\State \textbf{Input parameters:}
    \State \hskip2em $\mathsf{Good\_Group}$ $\mathit{G}_i$ \BlueComment{see \Cref{definition:good_group}}
    \State \hskip2em $\mathsf{Prediction}$ $\mathit{a}_i$ \BlueComment{$a_i$ is $p_i$'s prediction string}

\medskip
\State $\mathsf{SimpleElection}(\mathit{G}_i, \mathit{a}_i)$:
    \State \hskip2em \textbf{let} $\mathit{v}_i$ be a binary string where $\mathit{v}_i[j] = \mathit{a}_i[\mathit{G}_i[j]]$ for each $1 \le j \le |\mathit{g}_i|$ 
    \State \hskip2em \textbf{broadcast} $\langle \textsc{vote}, \mathit{v}_i\rangle$

    \smallskip
    \State \hskip2em \textbf{let} $j$ be the smallest index such that $p_i$ received more than $n/2$ $\langle \textsc{vote}, \mathit{v} \rangle$ with $\mathit{v}[j] = 1$
    \State \hskip2em \textbf{return} $\mathit{G}_i[j]$
\end{algorithmic}
\end{algorithm}

We can verify that \Cref{algorithm:simple-group-leader} solves the group leader election.

\begin{lemma}\label{lemma:simple_group_leader}
    There is a group leader election algorithm for a good group $G$. The algorithm terminates in $1$ round, exchanges $\mathcal{O}(n^2)$ messages, and exchanges $\mathcal{O}(n^2|G|)$ bits.
\end{lemma}
\begin{proof}
    As $G$ is a good group, no process is misclassified. Thus, each process will obtain the same index $j$. Moreover, as $G$ is a good group, $\mathit{G}_i[j]$ must be an honest process. Finally, the complexities follow from the algorithm.
\end{proof}

The communication complexity of \Cref{algorithm:simple-group-leader} is upper bounded by $\mathcal{O}(n^3)$, i.e., when the group size is $\mathcal{O}(n)$. In the following sections, we will describe protocols solving the same problem, but with a better communication complexity, i.e., $\mathcal{O}(n^2 \log n)$ or $\mathcal{O}(n^2 \kappa)$.

Next, we state several key combinatorial lemmas on the existence of good groups that are crucial for our results. We start by examining how the processes are grouped. Roughly speaking, each group has a similar size, and a single process belongs to a constant number of groups.

\begin{definition}[$m$-Grouping]
    \label{def:m_grouping}
    A grouping of all processes into $m$ non-empty groups is said to be $m$-grouping if each group has $\Theta(n/m)$ processes and each process is part of $\mathcal{O}(1)$ group.
\end{definition}

In this paper, we only consider $m$-grouping with $n \ge m$ that can be computed deterministically. In particular, our results only require all processes to compute the same $m$-grouping for each value of $m$.

We are now ready to state the key lemmas. Intuitively, when $B = o(n^2)$, only  $o(n)$ processes are misclassified or weakly misclassified (i.e., this is the regime where we want our predictions to be useful). We will choose the number of groups to be proportional to the number of such processes. In the authenticated setting, our grouping is simple and is based on a partitioning of the processes.

\begin{restatable}{lemma}{existgood}
    \label{lemma:exist_1/2_good_groups}
    Suppose there are $f < (\frac{1}{2}-\epsilon)n$ faulty processes and $k$ misclassified processes, where $0 < \epsilon < \frac{1}{2}$ is a constant. We consider an $m$-grouping obtained by arbitrarily partitioning the $n$ processes into $m$ disjoint groups such that each group has size $\lfloor n/m \rfloor$ or $\lceil n/m \rceil$. Then, there exist two positive constants $c_1$ and $c_2$ such that, if $c_1k < m < c_2n$, this $m$-grouping has at least one good group.
\end{restatable}

In the unauthenticated setting, such a simple partitioning is not enough. We now give some ideas as to why the following grouping makes sense. First, we remark that in the unauthenticated setting, less than $1/3$ of the processes are Byzantine. Meanwhile, for resilient groups, we only need less than $1/2$ of the processes in each group to be Byzantine. This gap between $1/3$ and $1/2$ can be used by an extractor-based structure to obtain an $m$-grouping where an arbitrarily small proportion of the groups are not resilient, while each process is only a part of $\mathcal{O}(1)$ groups. Specifically, using the extractor-based committee assignment from \cite{dufay2026optimalca}, we obtain the following.

\begin{restatable}{lemma}{existresilient}
    \label{lemma:exist_resilient}
     Suppose there are $f < n/3$ faulty processes and $k$ weakly misclassified processes. Then, there exist some constants $D \geq 1$ and $c > 0$ such that for any $1 \leq m \leq n$, there exists an $m$-grouping where each process is part of at most $D$ groups such that if $m \geq c \cdot k$, then at least $5m/6$ of the groups are resilient. Moreover, this $m$-grouping can be computed deterministically in polynomial time.
\end{restatable}

We remark that in the proof of the lemma above, the grouping is based on a multigraph. This means that a process may appear multiple times in a given group (at most $D = \mathcal{O}(1)$ times). We explain in \Cref{subsection:unauthenticated:election} how this is handled during leader election.

We defer the proofs of \Cref{lemma:exist_1/2_good_groups} and \Cref{lemma:exist_resilient} to \Cref{section:missing_proofs_using_pred}.

%% file: sections/unauthenticated.tex
\section{Unauthenticated Algorithm}\label{section:unauthenticated}

Our unauthenticated algorithm runs two algorithms in parallel:
(i) a Byzantine agreement protocol with a round complexity that scales with how accurate the predictions are; and (ii) an early-stopping Byzantine agreement with a round complexity that scales with how many Byzantine processes there are. Our primary focus is on implementing the first one. 

At a high level, the key idea is to use Byzantine agreement on an \emph{implicit} committee, where each process chooses a set of processes they think are the committee. By using $m$-grouping and group leader elections for choosing those processes, in resilient groups with few prediction errors, each process will choose the same honest process, which will constitute the implicit committee. Moreover, by setting $m$ such that there will be many resilient groups with successful elections (and so, bigger implicit committee), the algorithm will correctly solve Byzantine agreement when there are not too many weakly misclassified processes.
After that, we also show that there is a group leader election which, combined with the Byzantine agreement with an implicit committee and the approach from \cite{bendavid2025predictions}, yields Byzantine agreement that terminates in $\mathcal{O}(\min\{B/n, f\})$ rounds and sends $\tilde{\mathcal{O}}(n^{2})$ bits.

We first describe the Byzantine agreement with an implicit committee in \Cref{subsec:ba_implicit}. Next, we present the group leader election that is crucial for our agreement that exchanges $\tilde{\mathcal{O}}(n^{2})$ bits in \Cref{subsection:unauthenticated:election}. Finally, we describe how they are used along with the approach from \cite{bendavid2025predictions} in \Cref{subsection:unauthenticated:ba_predictions}.

\input{sections/sub/agreement_implicit}

\input{sections/sub/unauth_group_leader_election}

\subsection{Full Algorithm}\label{subsection:unauthenticated:ba_predictions}

In our unauthenticated algorithm for Byzantine agreement with classification predictions, we adopt the same general ``guess-and-double'' structure as \cite{bendavid2025predictions} (see \Cref{algorithm:unauth_ba}). Let us describe the approach from a high-level perspective. The algorithm goes through phases. In each phase, the algorithm estimates $\hat{k}$, the maximum number of actual faults and weakly misclassified processes.\footnote{Regarding the number of weakly misclassified processes, we are actually estimating the prediction errors $B$. However, this is fine as by estimating $B$, we are also estimating the number of weakly misclassified processes, which is $\mathcal{O}(B/n)$.} The algorithm then runs an \emph{early-stopping} Byzantine agreement sub-protocol that terminates in $\mathcal{O}(f)$ rounds (with $f$ actual faults), and a Byzantine agreement sub-protocol that, when there are at most $k$ weakly misclassified processes, terminates in $\mathcal{O}(k)$ rounds. Then, using the estimate $\hat{k}$ from before, we abort both agreement instances after $\mathcal{O}(\hat{k})$ rounds, (line~\ref{line:unauth_full:early_ba} and line~\ref{line:unauth_full:implicit}). 

Graded consensus instances are then used to ensure safety while adopting the output of each Byzantine agreement. Formally, graded consensus is a weaker consensus primitive where each honest process $p_i$ proposes its input $v_i$, and outputs a decision $d_i$ with a grade $g_i \in \{0, 1\}$ such that the followings are satisfied:
\begin{itemize}
    \item \emph{Consistency}: if an honest process decides $v^*$ with grade $1$, then all honest processes decide $v^*$ (potentially with grade $0$).
    \item \emph{Strong unanimity}: if each honest process proposes the same value $v$, then they all decide $v$ with grade $1$.
    \item \emph{Termination}: each honest process eventually decides.
\end{itemize}

Each process runs a graded consensus before invoking a Byzantine agreement sub-protocol (line~\ref{line:unauth_full:gc_1} and line~\ref{line:unauth_full:gc_2}), and only adopts the output of the agreement if it decides with grade $0$ from the graded consensus; otherwise, it commits to the output of the graded consensus (line~\ref{line:unauth_full:gc_adopt_1} and line~\ref{line:unauth_full:gc_adopt_2}). Lastly, termination is guaranteed when the estimate $\hat{k}$ is big enough, as at least one of the two agreement sub-protocols will terminate fast enough (and correctly). By doubling the $\hat{k}$ in each phase, the algorithm will terminate in $\mathcal{O}(\min\{B/n, f\})$ rounds.

We now focus on our agreement sub-protocol that terminates in $\mathcal{O}(k)$ rounds when there are at most $k$ weakly misclassified processes (which is the new contribution in this paper). Here, we use the Byzantine agreement with an implicit committee from \Cref{subsec:ba_implicit}. 
In a phase, we set the number of groups $m \in \Theta(\hat{k})$ such that if $\hat{k} \ge cB/n$ for some constant $c$, there will be at least $\frac{5m}{6}$ resilient groups (for \Cref{thm:unauth_subcubic_bits}) at line~\ref{line:unauth_full:pick_m}; such $m$ exists due to \Cref{lemma:exist_resilient}. Then, each process $p_i$ deterministically computes the $m$-grouping $G$. For each computed group, each process participates in a group leader election from \Cref{subsection:unauthenticated:election} to produce a vector $\mathit{L}_i$; these can be run in parallel (line~\ref{line:unauth_full:election}). Then, using the vector of processes $\mathit{L}_i$, each process invokes the Byzantine agreement with implicit committee (line~\ref{line:unauth_full:implicit}). Therefore, the sub-protocol will terminate in $\mathcal{O}(m) = \mathcal{O}(\hat{k})$ rounds.
\begin{algorithm} [!ht]
\caption{Unauthenticated Byzantine Agreement with Classification Predictions: Pseudocode (for process $p_i$)}
\label{algorithm:unauth_ba}
\begin{algorithmic} [1]
\footnotesize

\State \textbf{Input parameters:}
    \State \hskip2em $\mathsf{Value}$ $\mathit{v}_i$ 
    \State \hskip2em $\mathsf{Prediction}$ $\mathit{a}_i$

\smallskip
\State \textbf{Local variables:}
    \State \hskip2em $\mathsf{Grade}$ $\mathit{g}_i \gets \bot$
    \State \hskip2em $\mathsf{Value}$ $\mathit{decision}_i \gets \bot$
    
\medskip
\State $\mathsf{UnauthenticatedAgreement}(\mathit{v}_i, \mathit{a}_i)$:
    \State \hskip2em \textbf{for} $\mathsf{Integer}$ $\phi \gets 1$ to $\lceil \log_2(t) \rceil+1$: \BlueComment{at most $\hat{k} = 2^{\phi-1}$ faults or weakly misclassified processes}
        \State \hskip4em $(\mathit{v}_i, \mathit{g}_i) \gets \mathsf{GradedConsensus}(\mathit{v}_i)$
        \State \hskip4em \textbf{let} $T \gets \alpha \cdot 2^{\phi-1}$
        \State \hskip4em \textbf{let} $v' \gets \mathsf{ByzantineAgreement}(\mathit{v}_i)$ executed for $T$ rounds\label{line:unauth_full:early_ba} \BlueComment{early-stopping agreement}
        \State \hskip4em \textbf{if} $\mathit{g}_i = 0$:\label{line:unauth_full:gc_adopt_1}
            \State \hskip6em $\mathit{v}_i \gets v'$

        \smallskip
        \State \hskip4em $(\mathit{v}_i, \mathit{g}_i) \gets \mathsf{GradedConsensus}(\mathit{v}_i)$\label{line:unauth_full:gc_1}
        \State \hskip4em \textbf{let} $m \gets \beta \cdot 2^{\phi-1}$\label{line:unauth_full:pick_m} \BlueComment{$m \in \Theta(\hat{k})$ such that there are $\ge \frac{5m}{6}$ resilient groups}
        \State \hskip4em \textbf{let} $L_i \gets [p_1, p_2, \dots, p_m]$
        \State \hskip4em \textbf{let} $\mathit{G}_i \gets \mathsf{M\_Grouping}(\Pi, m)$ \BlueComment{see \Cref{def:m_grouping}}
        \State \hskip4em \textbf{for} $\mathsf{Integer}$ $j \gets 1$ to $m$: \BlueComment{run in parallel}
            \State \hskip6em $\mathit{L}_i[j] \gets \mathsf{GroupLeaderElection}(\mathit{G}_i[j], \mathit{a}_i)$\label{line:unauth_full:election}
        \State \hskip4em \textbf{let} $v' \gets \mathsf{AgreementWithImplicitCommittee}(\mathit{v}_i, \mathit{L}_i)$ \label{line:unauth_full:implicit} \BlueComment{\Cref{algorithm:ba-groups}}
         \State \hskip4em \textbf{if} $\mathit{g}_i = 0$:\label{line:unauth_full:gc_adopt_2}
            \State \hskip6em $\mathit{v}_i \gets v'$
    
        \smallskip
        \State \hskip4em $(\mathit{v}_i, \mathit{g}_i) \gets \mathsf{GradedConsensus}(\mathit{v}_i)$\label{line:unauth_full:gc_2}
        \State \hskip4em \textbf{if} $\mathit{decision}_i \ne \bot$:
            \State \hskip6em \textbf{return} $\mathit{decision}_i$ and \textbf{halt}
        \State \hskip4em \textbf{else if} $\mathit{g}_i = 1$:
            \State \hskip6em $\mathit{decision}_i \gets \mathit{v}_i$ \BlueComment{will decide in the next iteration}

    \State \hskip2em \textbf{return} $\mathit{decision}_i$ \BlueComment{if has not halted yet}
\end{algorithmic}
\end{algorithm}

In \cite{bendavid2025predictions}, it was proven that the overall algorithm correctly solves Byzantine agreement if the Byzantine agreement protocol with predictions works as claimed. The graded consensus and early-stopping Byzantine agreement are sufficient to prove safety and liveness, while the prediction-based Byzantine agreement is used to ensure the improved round complexity. Therefore, our main challenge is to analyze the round and communication complexities of the algorithm. We start by focusing on the number of phases run until all honest processes decide.

\begin{lemma}[Restated from \cite{bendavid2025predictions}]
    Suppose at phase $\phi$, the output of the early-stopping Byzantine agreement or the Byzantine agreement with implicit committee satisfies both strong unanimity and agreement. Then, all honest processes decide by the end of phase $\phi+1$.
\end{lemma}

\begin{lemma}
    \label{lemma:unauth_ba:earliest_pred_phase}
    Suppose there are $k$ weakly misclassified processes. Then, in a phase $\phi$ with $2^{\phi-1} > \max\{k, B/n\}$, the output of the Byzantine agreement with implicit committee satisfies both strong unanimity and agreement.
\end{lemma}
\begin{proof}
    Recall that the grouping is done using \Cref{lemma:exist_resilient}. Therefore, there exists some constant $D \geq 1$ such that every party is part of at most $D$ groups for any grouping computed using this approach. Moreover, \Cref{lemma:exist_resilient} also guarantees that there exists a constant $c > 0$ such that if $m \geq c \cdot k$, then at least $5m/6$ groups are resilient. Thus, we perform the grouping using \Cref{lemma:exist_resilient} and $m = \max\{c, 6D\} \cdot 2^{\phi - 1} = \Theta(2^{\phi - 1})$. 
    
    Assume that $ \max\{k, B/n\} < 2^{\phi-1}$. We want to show that the Byzantine agreement with implicit committees will succeed. Because of \Cref{lemma:black_box_ba}, we know that it will be the case if at least strictly more than $2m/3$ groups elect an honest leader. As seen above, at least $5m/6$ groups are resilient. Moreover, using \Cref{lemma:unauth-misprediction}, if a resilient group does not elect an honest leader, then it contains at least $n$ prediction errors. We note that there are $B$ prediction errors in total, and because a process appears in at most $D$ groups, a prediction error can only be ``shared'' by at most $D$ groups. Therefore, at most $DB/n$ resilient groups can fail their leader election. We note that $m \geq 6D 2^{\phi - 1}$ and we assume that $2^{\phi - 1} > B/n$. So strictly less than $m/6$ resilient groups will fail their leader election. Thus, more than $5m/6 - m/6 = 4m/6 = 2m/3$ groups elect an honest leader. The assumptions needed by the Byzantine agreement with implicit committee are satisfied. Therefore, the output of the Byzantine agreement with implicit committee will satisfy both strong unanimity and agreement.
\end{proof}


The earliest phase where the output of the early-stopping Byzantine agreement satisfies both strong unanimity and agreement is at most $\log_2(f) + \mathcal{O}(1)$ (when it is run for long enough). Then, from \Cref{lemma:unauth_ba:earliest_pred_phase}, the earliest phase where the output of the Byzantine agreement with implicit committee satisfies both properties is at most $\log_2(B/n) + \mathcal{O}(1)$. Therefore, we get the following.

\begin{lemma}
    \label{lemma:unauth_full:max_phase}
    All honest processes decide within $\log_2 (\min\{B/n, f\}) + \mathcal{O}(1)$ phases.
\end{lemma}

Next, to prove our main unauthenticated result, we use the following implementation for our graded consensus and early-stopping Byzantine agreement.

\begin{lemma}[Graded consensus, restated from \cite{civit2024efficient}]
    \label{lemma:gc}
    There is an unauthenticated algorithm for graded consensus that tolerates up to $t < n/3$ faulty processes. The algorithm terminates in $\mathcal{O}(1)$ rounds, exchanges $\mathcal{O}(n^2)$ messages, and exchanges $\mathcal{O}(n^2)$ bits.
\end{lemma}

\begin{lemma}[Early-stopping Byzantine agreement, restated from \cite{lenzen2022phase}]
    \label{lemma:early_ba}
    There is an unauthenticated algorithm for Byzantine agreement that tolerates up to $t < n/3$ faulty processes. The algorithm terminates in $\mathcal{O}(f)$ rounds (where $f \le t$ is the actual number of faults), exchanges $\mathcal{O}(n^2)$ messages, and exchanges $\mathcal{O}(n^2)$ bits.
\end{lemma}




We are now ready to prove our main unauthenticated result. 

\unauthsubcubic*
\begin{proof}
    
    From \Cref{lemma:exist_resilient}, \Cref{lemma:unauth:election} (group leader election), \Cref{lemma:gc} (graded consensus), and \Cref{lemma:early_ba} (early-stopping Byzantine agreement), the algorithm tolerates up to $t < n/3$ Byzantine processes.

    From \Cref{lemma:ba_implicit_complexities}, when our Byzantine agreement with implicit committee is run on $\mathcal{O}(2^{\phi})$ groups, it will terminate in $\mathcal{O}(2^{\phi})$ rounds, exchange $\mathcal{O}(n^2)$ messages, and exchange $\mathcal{O}(n^2 \log n)$ bits.
    Then, observe that $\sum_{g \in G} |g| = \mathcal{O}(n)$ for each computed $m$-grouping $G$ in any phase, following the definition of a $m$-grouping (a process appears in $\mathcal{O}(1)$ groups). Thus, in each phase, honest processes exchange $\mathcal{O}(n^2)$ messages and $\mathcal{O}(n^{2} \log(n))$ bits across all group leader elections in $\mathcal{O}(1)$ rounds. Combined with \Cref{lemma:gc} and \Cref{lemma:early_ba}, each phase $\phi$ terminates in $\mathcal{O}(2^{\phi})$ rounds, exchanging $\mathcal{O}(n^2)$ messages and $\mathcal{O}(n^2 \log n)$ bits. Finally, following \Cref{lemma:unauth_full:max_phase}, we obtain the stated complexities.
\end{proof}

%% file: sections/sub/agreement_implicit.tex
\subsection{Byzantine Agreement with an Implicit Committee}\label{subsec:ba_implicit}

In our agreement with an implicit committee, besides their proposal, each process $p_i$ is also given a vector of processes $\mathit{L}_i$. Processes may input different vectors of processes. However, we assume that (1) they all have the same size, (2) for each $p_i$, there is at most $\mathcal{O}(1)$ $j$ with $\mathit{L}_i[j] = p_i$, and (3) there are more than $\frac{2|\mathit{L}|}{3}$ indices $j$, such that $\mathit{L}_i[j]$ is the same honest process for each honest process $p_i$. Let us note that the third assumption is where the `implicit committee' comes in, as there exists a committee (common set), but processes do not know who they are. In our full algorithm, the first two assumptions are due to the $m$-grouping, while the last assumption will be due to the group leader election on each group.

Before describing the algorithm, let us remark on the following existing Byzantine agreement protocol that our algorithm relies on.

\begin{lemma}[Restated from \cite{berman1992bit}]
    \label{lemma:black_box_ba}
    There is an unauthenticated Byzantine agreement algorithm that tolerates up to $t < n/3$ Byzantine processes. The algorithm terminates in $\mathcal{O}(n)$ rounds, and each honest process sends $\mathcal{O}(n)$ bits.
\end{lemma}

Importantly, the complexities hold even if there are more than $t$ faults (although processes may then disagree).

We present our protocol in \Cref{algorithm:ba-groups}. If a process $p_i$ considers itself a part of the committee ($p_i \in \mathit{L}_i$), it then participates in the Byzantine agreement protocol from \Cref{lemma:black_box_ba}. Here, $p_i$ participates assuming there are $|\mathit{L}_i|$ processes, where the $j$-th process is $\mathit{L}_i[j]$ (line~\ref{line:ba_implicit:black_box}). After that, $p_i$ broadcasts its decision at line~\ref{line:ba_implicit_broadcast}. Finally, every process $p_i$ decides using the majority value it receives from the processes in its $\mathit{L}_i$.

We remark that we allow a single process to participate multiple times (at most $\mathcal{O}(1)$ times) in the committee. To handle this case, given a process $p_i$, for each index $j$ where $L_i[j] = p_i$, $p_i$ runs an instance of the protocol where it acts as the $j$-th process. Therefore, each process may run up to $\mathcal{O}(1)$ instances of \Cref{lemma:black_box_ba} in parallel. Importantly, due to the possibility that a process may `simulate' multiple processes in the executed Byzantine agreement, each message also needs to specify who the `simulated' sender and receiver are.

Thus, during the execution of the Byzantine agreement in \Cref{lemma:black_box_ba}, we add two identifiers in each message sent: one giving the index of the sender and one giving the index of the receiver, both between $1$ and $|L_i|$. This way, processes know which instances the message is from and to. If the implicit sender/receiver of a message does not match the identifier given in the message, it is discarded. Finally, note that using this approach, each message has an additional $\mathcal{O}(2 \lceil \log |L_i| \rceil) = \mathcal{O}(\log n)$ bits.

\begin{algorithm} [h]
\caption{Byzantine Agreement with Implicit Committee: Pseudocode (for process $p_i$)}
\label{algorithm:ba-groups}
\begin{algorithmic} [1]
\footnotesize
    
\smallskip
\State \textbf{Input parameters:}
    \State \hskip2em $\mathsf{Vector}(\mathsf{Process})$ $\mathit{L}_i$
    \State \hskip2em $\mathsf{Value}$ $\mathit{v}_i$

\medskip
\State $\mathsf{AgreementWithImplicitCommittee}(\mathit{v}_i, \mathit{L}_i)$:
    \State \hskip2em \textbf{if} $p_i \in \mathit{L}_i$:
        \State \hskip4em \textbf{let} $\mathit{d}_i \gets \mathsf{ByzantineAgreement}(\mathit{v}_i)$ over $\Pi_i' = \{p'_1, \dots, p'_{|L_i|}\}$ with $p'_j = L_i[j]$\label{line:ba_implicit:black_box} \BlueComment{\Cref{lemma:black_box_ba}}
        \State \hskip4em \textbf{broadcast} $\langle \textsc{decision}, \mathit{d}_i \rangle$\label{line:ba_implicit_broadcast} 

    \smallskip
    \State \hskip2em \textbf{let} $\mathit{decision}_i$ be the majority value of $\textsc{decision}$ messages from $\mathit{L}_i$
    \State \hskip2em \textbf{return} $\mathit{decision}_i$
\end{algorithmic}
\end{algorithm}

We now prove the guarantees of the algorithm.

\begin{lemma}
    The algorithm correctly solves Byzantine agreement.
\end{lemma}
\begin{proof}
    The execution of the Byzantine agreement from \Cref{lemma:black_box_ba} will be equivalent to the one with $\Pi^* = \{p^*_1, \dots, p^*_{|L|}\}$ where for each $1 \le j \le |L|$:
    \begin{itemize}
        \item $p^*_j$ is honest if $\mathit{L}_i[j]$ is the same honest process for each honest process $p_i$, and
        \item $p^*_j$ is Byzantine otherwise.
    \end{itemize}
    In other words, those in the implicit committee are honest, and the rest are Byzantine. As there are more than $\frac{2|\mathit{L}|}{3}$ indices $j$ where $\mathit{L}_i[j]$ is the same honest process for each honest process $p_i$, there will be less than $\frac{|L|}{3}$ Byzantine processes. Therefore, the output of processes in the implicit committee will satisfy agreement and strong unanimity. Moreover, the value they broadcast at line~\ref{line:ba_implicit_broadcast} will be the majority value received by any honest process. Hence, each honest process will agree on a value that satisfies strong unanimity.
\end{proof}

\begin{lemma}
    \label{lemma:ba_implicit_complexities}
    The algorithm terminates in $\mathcal{O}(|L|)$ rounds, exchanges $\mathcal{O}(n^2)$ messages, and exchanges $\mathcal{O}(n^2 \log n)$ bits.
\end{lemma}
\begin{proof}
    There is only one additional round besides the invoked Byzantine agreement, hence the round complexity follows from \Cref{lemma:black_box_ba}. From the assumptions, \Cref{lemma:black_box_ba}, adding additional $\mathcal{O}(\log n)$ bits in each message in \Cref{lemma:black_box_ba}, and line~\ref{line:ba_implicit_broadcast}, each honest process sends $\mathcal{O}(n)$ messages and $\mathcal{O}(n \log n)$ bits.
\end{proof}

Furthermore, due to \Cref{lemma:black_box_ba}, the complexities hold even if there are fewer than $\frac{2|\mathit{L}|}{3}$ processes in the implicit committee. This will be important later, as we invoke this algorithm.

%% file: sections/sub/unauth_group_leader_election.tex
\subsection{Group Leader Election}\label{subsection:unauthenticated:election}

We now focus on an efficient algorithm for group leader election in resilient groups (see \Cref{definition:resilient_group}). Our goal is to prove the following lemma. We define the number of prediction errors in a group as the sum of the prediction errors over all processes in this group.

\begin{restatable}{lemma}{unauthelection}
    \label{lemma:unauth:election}
    Consider a resilient group $G$ with fewer than $n$ prediction errors. Then there is a group leader election algorithm assuming $f < n/3$. The algorithm terminates in $\mathcal{O}(1)$ rounds, exchanges $\mathcal{O}(n|G|)$ messages, and $\mathcal{O}(n|G| \log (n))$ bits.
\end{restatable}

Let us stress that the lemma works for any group size $|G|$. Also, the round, message, and bit complexity holds even if the requirement on the group being resilient or on the number of prediction errors is not satisfied.

The algorithm is presented at \Cref{algorithm:group-leader-big}. At a high level, all processes send (some part of) their prediction concerning $G$ to every process in $G$. Then every process in $G$ uses this information to decide a leader and broadcasts this decision to everyone. Because we assume $G$ is resilient, it has a majority of honest parties. Therefore, if all honest parties decide on the same honest party as their leader, the leader election will be successful as long as processes decide the leader which was elected the most.

One key idea is that instead of sending the full $|G|$-bit prediction list, one party can just send the $10$ parties with the smallest index in $G$ it predicted as honest. We show that as long as $G$ does not have too many prediction errors, this approach will still be successful while improving the communication complexity by a factor of $\frac{|G|}{\log n}$ compared to the naive approach. Meanwhile, when this approach fails, we show that there will be sufficiently many prediction errors, which later allows us to limit how many elections could fail.

We note that in the full algorithm, the grouping is done via \Cref{lemma:exist_resilient}. Thus, $G$ is a multiset, i.e., a party may appear multiple times in a group. This can be handled simply by weighting the leader messages from a process on line~\ref{line:unauth-elect:decide} of \Cref{algorithm:group-leader-big} by its multiplicity in $G$. We omit this in the exposition for clarity.

\begin{algorithm} [h]
\caption{Group Leader Election in the Unauthenticated Setting: Pseudocode (for process $p_i$)}
\label{algorithm:group-leader-big}
\begin{algorithmic} [1]
\footnotesize

\State \textbf{Input parameters:}
    \State \hskip2em $\mathsf{Resilient\_Group}$ $\mathit{G}$ \BlueComment{see \Cref{definition:resilient_group}}
    \State \hskip2em $\mathsf{Prediction}$ $\mathit{a}_i$

\medskip
\State $\mathsf{UnauthenticatedElection}(\mathit{G}, \mathit{a}_i)$:
    \State \hskip2em \textbf{let} $\mathit{V}_i$ be the set of index $j$ with $p_j \in \mathit{G}$ and $\mathit{a}_i[j] = 1$
    \State \hskip2em \textbf{if} $|\mathit{V}_i| > 10$, only keep the smallest $10$ indices in $\mathit{V}_i$
    \State \hskip2em \textbf{send} $\langle \textsc{vote}, \mathit{V}_i\rangle$ to each $p_j \in \mathit{G}$

    \smallskip
    \State \hskip2em \textbf{if} $p_i \in \mathit{G}$:
        \State \hskip4em \textbf{let} $\ell_i$ be the smallest index $j$ in $\mathit{G}$ where $p_i$ received more than $n/2$ $\langle \textsc{vote}, \mathit{V} \rangle$ with $j \in \mathit{V}$
        \State \hskip4em \textbf{if} $\ell_i$ exists:
            \State \hskip6em \textbf{broadcast} $\langle \textsc{leader}, \mathit{p_{\ell_i}} \rangle$
        \State \hskip4em \textbf{else}:
            \State \hskip6em \textbf{let} $q$ be an arbitrary process in $G$
            \State \hskip6em \textbf{broadcast} $\langle \textsc{leader}, \mathit{q} \rangle$

     \smallskip
    \State \hskip2em \textbf{let} $p_\ell \in \mathit{G}$ be a process such that $p_i$ received $\langle \textsc{leader}, p_\ell \rangle$ from more than $|\mathit{G}|/2$ processes in $\mathit{G}$ \label{line:unauth-elect:decide}
    \State \hskip2em \textbf{if} $p_\ell$ exists:
        \State \hskip4em \textbf{return} $p_\ell$
    \State \hskip2em \textbf{else}:
        \State \hskip4em \textbf{let} $q$ be an arbitrary process in $G$
        \State \hskip4em \textbf{return} $q$
\end{algorithmic}
\end{algorithm}

We now prove the guarantees of the algorithm, assuming all honest processes input the same group $\mathit{G}$.

\begin{lemma}\label{lemma:unauth-misprediction}
    Assume $G$ is resilient. If a common honest leader from $G$ is not elected, then $G$ has at least $n$ prediction errors.
\end{lemma}

\begin{proof}
    If $G$ is resilient, then it has a majority of honest processes. Let $q$ be the honest process in $G$ with the smallest index. We want to show that if not all honest processes consider $q$ to be elected, then $G$ has at least $n$ prediction errors.

    We remark that, because $G$ has an honest majority, if all honest processes in $G$ designate $q$ as the leader, it will be chosen as such by every honest process. Therefore, we assume that an honest process $p \in G$ did not choose $q$ as the leader, and we want to prove that $G$ has at least $n$ mispredictions.

    If $p$ does not choose $q$ as its leader, then either (1) it received at least $n/2$ votes for a party with a smaller index than $q$ or (2) it did not receive at least $n/2$ votes for $q$. 
    
    In the first case, because we chose $q$ as the honest process in $G$ with the smallest index, it means that a Byzantine process $u \in G$ received at least $n/2$ votes. Because $u$ is Byzantine, it can only receive votes from the $f$ Byzantine parties or $b_u$ honest parties which made a prediction error regarding $u$. However, $G$ is resilient, so $u$ is not weakly misclassified, implying that it is not misclassified, so $f + b_u < n/2$. Therefore, $u$ could not have received $n/2$ votes.

    In the second case, because $G$ is resilient, it implies that $q$ is not weakly misclassified. Let $A$ be the set of honest processes which predict $q$ as honest. Because $q$ is honest and not weakly misclassified, we have $|A| \geq n/2 + n/10$. Let $C \subseteq A$ be the processes in $A$ which did not send a vote for $q$. Because we assume $p$ did not receive $n/2$ votes for $q$, it implies that $|C| \geq |A| - n/2 \geq n/10$. We now consider a party $u \in C$. We know that $u$ is honest, predicted $q$ as honest, and yet did not vote for $q$. This means that it voted for $10$ other processes in $G$ with smaller indices that $u$ predicted as honest. Because $q$ is the honest process in $G$ with the smallest index, this implies that these $10$ votes were for Byzantine parties that $u$ predicted as honest; these are prediction errors. Therefore, the total number of prediction errors in $G$ is at least $10 \cdot |C| \geq 10 \cdot n/10 = n$.
\end{proof}

By taking the contrapositive of \Cref{lemma:unauth-misprediction}, if there are not too many prediction errors in the group, a common honest leader will be elected.

\begin{corollary}\label{corr:unauth-misprediction}
    Assume $G$ is resilient. If $G$ has less than $n$ prediction errors, then a common honest leader from $G$ is elected.
\end{corollary}

Next, we prove the complexities of the algorithm.

\begin{lemma}
    \label{lemma:unauth-election-complexities}
    For any group $G$, $\mathsf{UnauthenticatedElection}$ has message complexity $\mathcal{O}(n|G|)$, bit complexity $\mathcal{O}(n |G| \log n)$ and terminates in $2$ rounds.
\end{lemma}

\begin{proof}
    This algorithm consists of two rounds. In the first round, all $n$ processes send a message to each process in $G$, and each message contains up to $10$ indicies, each taking $\lceil \log n \rceil$ bits. In the second round, every process in $G$ sends a leader index, taking $\lceil \log n \rceil$ bits to every process. The round, message, and bit complexities follow from this observation.
\end{proof}

Finally, \Cref{lemma:unauth:election} follows from \Cref{corr:unauth-misprediction} and \Cref{lemma:unauth-election-complexities}.

%% file: sections/authenticated.tex
\section{Authenticated Algorithm}\label{section:authenticated}

In this section, we describe our Byzantine agreement with classification predictions that uses cryptography. The algorithm can tolerate up to $t < (\frac{1}{2}-\epsilon)n$ Byzantine processes, terminates in $\mathcal{O}(\min\{B/n, f\})$ rounds, and exchanges $\mathcal{O}(n^2\kappa)$ bits. The algorithm is based on a well-known reduction from Byzantine agreement with strong unanimity to the one with \emph{external validity} (also known as \emph{validated Byzantine agreement}), where a decision $v$ must satisfy a pre-determined predicate $\extValid{\cdot}$ (that is, $\extValid{v} = \mathit{true}$). Importantly, the reduction takes constant rounds, and exchanges $\mathcal{O}(n^2\kappa)$ bits (see \Cref{section:sa_to_ev} for the details). Next, a key component in our result is efficient group leader election that, after a pre-processing step with constant rounds and $\mathcal{O}(n^2\kappa)$ bits, can do the election with zero communication. Then, by using a leader-based validated Byzantine agreement and alternating between round-robin leaders and leaders obtained from the predictions (via $m$-grouping and group leader election), we get an algorithm with the complexities as desired.

We first describe our validated Byzantine agreement in \Cref{subsection:authenticated:vba}. Next, we describe the group leader election we use in \Cref{subsection:authenticated:election}. Finally, we describe our full algorithm in \Cref{subsection:authenticated:ba_predictions}, where we use the predictions along with the validated Byzantine agreement to obtain the desired result.

\subsection{Validated Byzantine Agreement}\label{subsection:authenticated:vba}

In the validated Byzantine agreement problem, each process has to agree on a value $v$ that satisfies a pre-determined predicate $\extValid{\cdot}$. Here, each process proposes a value that is externally-valid, i.e., satisfies the predicate. For our use case, besides their proposal, each process $p_i$ also has a vector of processes $L_i$ as an input. Our algorithm will work in a leader-based manner, where $p_i$ will consider $L_i[j]$ as the leader for view (iteration) $j$. Note that we do not assume all processes input the same vector, and so, during a view, two honest processes may perceive a different process as their leader.

Our implementation (\Cref{algorithm:vba}) that tolerates up to $t < (\frac{1}{2}-\epsilon)n$ faults for some constant $0 < \epsilon < \frac{1}{2}$ is based on the agreement part of the Byzantine broadcast protocol from \cite{wan2023amortizedbb}. At a high level, they employ quorum-based agreement, where the resilience is strengthened by the use of \emph{expander graph}: a graph (that typically has a low degree) with a good expansion property. In their algorithm, they use the following $(n, 2\epsilon, 1-2\epsilon)$-expander graph from Momose and Ren \cite{momose2021optimal}.

\begin{lemma}[$(n, 2\epsilon, 1-2\epsilon)$-expander graph, restated from \cite{momose2021optimal}]
    \label{lemma:expander_graph}
    For any constant $0 < \epsilon < \frac{1}{2}$ and any positive integer $n$, there is a constant degree graph with $n$ vertices such that any set of $2\epsilon n$ vertices is connected to a set of at least $(1-2\epsilon)n$ vertices.
\end{lemma}

The graph is used to forward the messages from the leader, which, as we show later, is important in strengthening the resilience.

However, there is a subtle problem we need to deal with. As mentioned before, the expander graph is used to forward the messages from the leader. Previously, this was fine because in a view, each process had the same leader. However, as now processes may have different leaders, the proof would not work anymore. There is a quite simple fix. We define \emph{leader proof}, a cryptographic object that proves a certain process is considered a leader by sufficiently many honest processes. 

\begin{definition}[Leader Proof]
    For a process $p_\ell$ and a view $j$, a leader proof for $p_\ell$ during view $j$ is a threshold signature $\mu$ where $\tssVerify{n-t}{\mu, \langle \textsc{leader}, p_\ell, j\rangle} = \mathit{true}$. 
\end{definition}

Then, each process that has a leader proof is considered the leader(s) of the view, and has to attach the leader proof in their messages to prove that they are indeed qualified as the leader. This solves the subtle correctness problem we had, and also allows us to bound the communication complexity per view. The rest of the protocol is quite similar to the original version; we describe it again for completeness.

Before describing the protocol, let us describe the following two cryptographic objects that will be crucial for the correctness of the protocol.

\begin{definition}[Commit Certificate]
    \label{def:commit_cert}
    For a value $v$ and a view $j$, a commit certificate for $v$ during view $j$ is a threshold signature $\mu$ where $\tssVerify{n-t}{\mu, \langle \textsc{commit}, v, j\rangle} = \mathit{true}$. 
\end{definition}

\begin{definition}[Decision Proof]
    For a value $v$, a decision proof for $v$ is a threshold signature $\mu$ where $\tssVerify{n-t}{\mu, \langle \textsc{decide}, v\rangle} = \mathit{true}$. 
\end{definition}

Now, let us describe the protocol as presented in \Cref{algorithm:vba}. First, the protocol uses an expander graph $\mathcal{G}$ from \Cref{lemma:expander_graph}. Next, the protocol goes through $|L|$ views, where $|L|$ is the size of every $\mathit{L}_i$ input by an honest process $p_i$ (line~\ref{line:vba:views}). We now focus on how a view goes.

\begin{algorithm} [!ht]
\caption{Validated Byzantine Agreement: Pseudocode (for process $p_i$)}
\label{algorithm:vba}
\begin{algorithmic} [1]
\footnotesize
\State \textbf{Uses:}
    \State \hskip2em $\mathsf{ExpanderGraph}$, instance $\mathcal{G}$ \BlueComment{see \Cref{lemma:expander_graph}}

\smallskip
\State \textbf{Input parameters:}
    \State \hskip2em $\mathsf{Value}$ $\mathit{v}_i$ \BlueComment{externally-valid value}
    \State \hskip2em $\mathsf{Vector}(\mathsf{Process})$ $\mathit{L}_i$
    
\smallskip
\State \textbf{Local variables:}
    \State \hskip2em $\mathsf{Value}$ $\mathit{commit\_value}_i \gets \bot$
    \State \hskip2em $\mathsf{CommitCertificate}$ $\mathit{commit\_cert}_i \gets \bot$
    \State \hskip2em $\mathsf{Value}$ $\mathit{decision\_value}_i \gets \bot$
    \State \hskip2em $\mathsf{DecisionProof}$ $\mathit{decision\_proof}_i \gets \bot$
    \State \hskip2em $\mathsf{LeaderProof}$ $\mathit{leader\_proof}_i \gets \bot$ 

\medskip
\State $\mathsf{ValidatedAgreement}(\mathit{v}_i, \mathit{L}_i)$:
    \State \hskip2em \textbf{for} $\mathsf{Integer}$ $j \gets 1$ to $|\mathit{L}_i|$:\label{line:vba:views}
        \State \hskip4em \textbf{let} $\ell_i \gets \mathit{L}_i[j]$
        \State \hskip4em \textbf{if} $\mathit{commit\_cert}_i \ne \bot$:
            \State \hskip6em \textbf{send} $\langle \textsc{val}, \mathit{commit\_value}_i, \mathit{commit\_cert}_i, \tssShareSign{i}{n-t}{\langle \textsc{leader}, \ell_i, j \rangle } \rangle$ to $\ell_i$ \label{line:vba:send_committed}
        \State \hskip4em \textbf{else}:
            \State \hskip6em \textbf{send} $\langle \textsc{val}, \mathit{v}_i, \tssShareSign{i}{n-t}{\langle \textsc{leader}, \ell_i, j \rangle } \rangle$ to $\ell_i$\label{line:vba:send_value}

        \smallskip
        \State \hskip4em \textbf{if} $p_i = \ell_i$ and received $n-t$ $\textsc{val}$ messages:
            \State \hskip6em $\mathit{leader\_proof}_i \gets \tssCombine{n-t}{\{psig\ |\ psig \text{ is received from } n-t \text{ processes}\}}$
            \State \hskip6em \textbf{let} $\mathit{v}, \mathit{cert}$ be the value with the highest commit view and its commit certificate from
            \State \hskip7em the received messages (if there is none, then $\mathit{v}$ is any valid value and $\mathit{cert} = \bot$)
            \State \hskip6em \textbf{broadcast} $\langle \textsc{propose}, \mathit{v}, \mathit{cert}, \mathit{leader\_proof}_i \rangle$\label{line:vba:propose}

        \smallskip
        \State \hskip4em \textbf{if} received $\langle \textsc{propose}, \mathit{v}, \mathit{cert}, \mathit{lp} \rangle$ from $\ell_i$ and it is from a view at least as high as $p_i$'s commit: \label{line:vba:check_view}
            \State \hskip6em \textbf{forward} to $p_i$'s neighbor in $\mathcal{G}$\label{line:vba:forward_value}
        \State \hskip4em \textbf{if} no conflicting value with a leader proof is received through $\mathcal{G}$: \label{line:vba:check_no_conflict}
            \State \hskip6em \textbf{send} $\langle \textsc{ack}, \mathit{v}, \tssShareSign{i}{n-t}{\langle \textsc{commit}, \mathit{v}, \mathit{j} \rangle} \rangle$ to $\ell_i$ \label{line:vba:send_partial_commit}
            \BlueComment{partially sign the value from leader}

        \smallskip
        \State \hskip4em \textbf{if} $p_i$ has a leader proof for view $j$ and $p_i$ received $n-t$ $\textsc{ack}$ messages for the same $\mathit{v}$:
            \State \hskip6em \textbf{let} $\mathit{cc} \gets \tssCombine{n-t}{\{psig\ |\ psig \text{ is received from } n-t \text{ processes}\}}$
            \State \hskip6em \textbf{broadcast} $\langle \textsc{commit}, \mathit{v}, \mathit{cc}, \mathit{leader\_proof}_i \rangle$\label{line:vba:broadcast_commit_cert}

        \smallskip
        \State \hskip4em \textbf{if} received $\langle \textsc{commit}, \mathit{v}, \mathit{cc}, \mathit{lp} \rangle$ from $\ell_i$:
            \State \hskip6em $\mathit{commit\_value}_i \gets \mathit{v}$; $\mathit{commit\_cert}_i \gets \mathit{cc}$
            \State \hskip6em \textbf{forward} to $p_i$'s neighbor in $\mathcal{G}$\label{line:vba:forward_commit}
            \State \hskip6em \textbf{send} $\langle \textsc{commit}, \mathit{commit\_value}_i, \tssShareSign{i}{n-t}{\langle \textsc{decide}, \mathit{commit\_value}_i \rangle} \rangle$ to $\ell_i$\label{line:vba:send_partial_decision}

        \smallskip
        \State \hskip4em \textbf{if} received valid $\langle \textsc{commit}, \mathit{v}, \mathit{cc}, \mathit{lp} \rangle$ through $\mathcal{G}$:
            \State \hskip6em $\mathit{commit\_value}_i \gets \mathit{v}$; $\mathit{commit\_cert}_i \gets \mathit{cc}$\label{line:vba:store_commit}
            
        \smallskip
        \State \hskip4em \textbf{if} $p_i$ has a leader proof for view $j$ and $p_i$ received $n-t$ $\textsc{commit}$ messages for the same $\mathit{v}$:
            \State \hskip6em \textbf{let} $\mathit{dp} \gets \tssCombine{n-t}{\{psig\ |\ psig \text{ is received from } n-t \text{ processes}\}}$
            \State \hskip6em \textbf{broadcast} $\langle \textsc{decide}, \mathit{v}, \mathit{dp}, \mathit{leader\_proof}_i \rangle$\label{line:vba:broadcast_decision}

        \smallskip
        \State \hskip4em \textbf{if} received $\langle \textsc{decide}, \mathit{v}, \mathit{dp}, \mathit{lp} \rangle$ from $\ell_i$:
            \State \hskip6em $\mathit{decision\_value}_i \gets \mathit{v}$; $\mathit{decision\_proof}_i \gets \mathit{dp}$\label{line:vba:store_decision}
    
    \State \hskip2em \textbf{return} $(\mathit{decision\_value}_i, \mathit{decision\_proof}_i)$\label{line:vba:decide}
\end{algorithmic}
\end{algorithm}

First, each process sends to its leader a partial signature for its leader proof, along with a value. This value is either a value with the highest commit certificate the process has (line~\ref{line:vba:send_committed}), or its input value if it does not have any commit certificate (line~\ref{line:vba:send_value}). In the former case, the process also sends the commit certificate. Next, if a process receives $n-t$ messages, it can produce a leader proof for itself. It does so, and then it checks whether it receives any value with a commit certificate. If it does, it proposes the value which commit certificate's view is the highest, and otherwise, it broadcasts any externally-valid value $v$ (line~\ref{line:vba:propose}); the process also attaches its leader proof in the message. Then, upon receiving a proposal for the value $v$ from its leader, a process forwards the message to its neighbors in $\mathcal{G}$ (line~\ref{line:vba:forward_value}). If no message with a different value is received through $\mathcal{G}$, then the process sends a partial signature for a commit certificate for the value $v$ proposed by its leader (line~\ref{line:vba:send_partial_commit}).

Next, if a process previously managed to produce its leader proof and it receives $n-t$ partial signatures for a commit certificate for a value $v$, it forms a commit certificate for $v$ and broadcasts it (line~\ref{line:vba:broadcast_commit_cert}). Then, upon receiving a value $v$ with its commit certificate from its leader, a process stores them, forwards the received message (that also contains a leader proof) to its neighbors in $\mathcal{G}$ (line~\ref{line:vba:forward_commit}), and then sends a partial signature for a decision proof for $v$ to its leader (line~\ref{line:vba:send_partial_decision}). Note that here, processes only disseminate the committed value and its certificate through $\mathcal{G}$, and always partially sign a decision proof (see line~\ref{line:vba:send_partial_decision} and line~\ref{line:vba:store_commit}). Then, similar to before, if a process has a leader proof and receives $n-t$ partial signatures for a decision for a value $v$, it forms a decision proof for $v$ and broadcasts it (line~\ref{line:vba:broadcast_decision}). Lastly, if a process receives a value accompanied by its decision proof from its leader, it stores them (line~\ref{line:vba:store_decision}) to be used as the final decision after all views have concluded (line~\ref{line:vba:decide}).

We now prove the correctness of the algorithm.

\begin{lemma}
    \label{lemma:vba:commit_cert_unique}
    If in a view, a commit certificate for the value $v$ and $v'$ exists, then $v = v'$.
\end{lemma}
\begin{proof}
    A commit certificate must be partially signed by at least $n-t-f \ge n-2t > 2\epsilon n$ honest processes. Honest processes that partially sign a commit certificate for $v$ must have forwarded the value $v$ through $\mathcal{G}$ at line~\ref{line:vba:forward_value}, which will be accepted by at least $(1-2\epsilon)n > 2t$ processes (note that the forwarded message must contain a valid leader proof). Hence, at least $2t+1-f$ honest processes would have received the forwarded value at line~\ref{line:vba:check_no_conflict}, in which they would not partially sign a commit certificate for a value $v' \ne v$. As at most $n-f-(2t+1-f) = n-(2t+1) < n-t-f$ honest processes could partially sign a commit certificate for a value $v' \ne v$, there would not be enough signatures to produce a commit certificate for $v' \ne v$.
\end{proof}

\begin{lemma}[Agreement]
    \label{lemma:vba:agreement}
    If a decision proof for the value $v$ and $v'$ exists, then $v = v'$.
\end{lemma}
\begin{proof}
    Let us focus on the first view where a decision proof exists, and suppose it is for the value $v$. At least $n-t-f> 2\epsilon n$ honest processes partially signed for a decision proof for the value $v$. These processes forwarded the commit certificate for $v$ through $\mathcal{G}$, which will be received by at least $(1-2\epsilon)n > 2t$ processes, out of which at least $2t+1-f$ are honest. By \Cref{lemma:vba:commit_cert_unique}, at this point, only a decision proof for $v$ may exist, and the commit certificate for the highest view is for $v$. Moreover, in the future views, at least $2t+1-f$ honest processes will only sign the commit certificate (and thus, decision proof) for the value $v$ (due to the check at line~\ref{line:vba:check_view}). Finally, no commit certificate nor decision proof for the value $v' \ne v$ can be produced (as there are not enough honest processes that would support it).
\end{proof}

\begin{lemma}[External Validity]
    \label{lemma:vba:external_validity}
    If a decision proof for the value $v$ exists, then $v$ satisfies external validity.
\end{lemma}
\begin{proof}
    If a decision proof for the value $v$ exists, then a commit certificate for the value $v$ exists. If a commit certificate for the value $v$ exists, either $v$ satisfies external validity ($\extValid{v} = \mathit{true}$) or a commit certificate for $v$ from a lower view exists. From induction, in the latter case, $v$ also satisfies external validity.
\end{proof}

Hence, safety is guaranteed. We then argue that when all honest processes have the same honest leader, then they all obtain a decision and its proof.

\begin{lemma}[Conditional Termination]
    \label{lemma:vba:conditional_termination}
    In a view $j$ where for each honest process $p_i$, $\mathit{L}_i[j] = p_\ell$ for some honest process $p_\ell$, each honest process obtains a decision and its decision proof.
\end{lemma}
\begin{proof}
    In such a view, $p_\ell$ will receive $\langle \textsc{val} \rangle$ messages from all honest processes. Thus, $p_\ell$ can build a leader proof for itself, and obtain a value with a commit certificate that all honest processes will accept. $p_\ell$ then broadcast them via $\langle \textsc{propose} \rangle$ message. Moreover, as there can be no process with a leader proof in that view besides $p_\ell$, all honest processes will send $\langle \textsc{ack} \rangle$ for the received value. Thus, $p_\ell$ can build a commit certificate for that value, and broadcast that value along with the commit certificate. Similarly, each honest process will accept this value and send a partial signature for its decision proof. Finally, $p_\ell$ builds a decision proof for that value and broadcast them to all honest processes.
\end{proof}

Therefore, as long as at least one of the views has the same honest leader, the algorithm correctly solves Byzantine agreement with external validity, given there are at most $t < (\frac{1}{2}-\epsilon)n$ faulty processes. Lastly, we prove the complexities of the algorithm. Specifically, we will focus on the complexities of each view.

\begin{lemma}
    \label{lemma:vba:leader_count}
    In each view, at most $\mathcal{O}(1)$ processes have a leader proof.
\end{lemma}
\begin{proof}
    Each leader proof needs $n-t$ partial signatures, so at least $n-t-f \ge n-2t$ came from honest processes. Thus, there are at most $\frac{n}{n-2t} < \frac{n}{2\epsilon n} = \frac{1}{2\epsilon} \in \mathcal{O}(1)$ such processes.
\end{proof}

\begin{lemma}
    \label{lemma:vba:complexities}
    Each view lasts for $\mathcal{O}(1)$ rounds. In each view, honest processes exchange $\mathcal{O}(n)$ messages and $\mathcal{O}(n\kappa)$ bits.
\end{lemma}
\begin{proof}
    The number of rounds can be observed from the algorithm. Then, observe that each message contains $\mathcal{O}(\kappa)$ bits. Furthermore, communications are done in a leader-to-all, all-to-leader pattern, along with the message forwarding through the expander graph. As there are $\mathcal{O}(1)$ leaders (\Cref{lemma:vba:leader_count}) and the expander graph has constant degree, $\mathcal{O}(n)$ messages and $\mathcal{O}(n\kappa)$ bits are exchanged in each view.
\end{proof}

\subsection{Group Leader Election}\label{subsection:authenticated:election}

We now present a simple way to do the group leader election in a good group (see \Cref{definition:good_group}) efficiently with cryptography (see \Cref{algorithm:group-leader-auth}). Here, we use partial signatures to vote for each process in the group that is predicted to be honest. When those processes have more than $n/2$ votes, they can use the (aggregated) votes to prove that more than $n/2$ processes predict that they are honest.\footnote{In the pseudocode, we use $\lceil (n+1)/2 \rceil$ votes to ensure the quantity is an integer. This does not affect our proofs.} In a good group, taking the smallest process with such a voting proof ensures each honest process will elect the same honest process as the leader. 

\begin{algorithm} [h]
\caption{Authenticated Group Leader Election: Pseudocode (for process $p_i$)}
\label{algorithm:group-leader-auth}
\begin{algorithmic} [1]
\footnotesize

\State \textbf{Input parameters:}
    \State \hskip2em $\mathsf{Good\_Group}$ $\mathit{G}_i$ \BlueComment{see \Cref{definition:good_group}}
    \State \hskip2em $\mathsf{Prediction}$ $\mathit{a}_i$

\medskip
\State $\mathsf{AuthenticatedElection}(\mathit{G}_i, \mathit{a}_i)$:
    \State \hskip2em \textbf{for} each $p_j \in \mathit{G}_i$ with $\mathit{a}_i[j] = 1$:
        \State \hskip4em \textbf{send} $\langle \textsc{vote}, p_j, \tssShareSign{i}{\lceil (n+1)/2 \rceil}{\langle \textsc{vote}, p_j \rangle} \rangle$ to $p_j$

    \smallskip
    \State \hskip2em \textbf{if} $p_i \in \mathit{G}_i$:
        \State \hskip4em \textbf{let} $vote\_proof_i \gets \tssCombine{\lceil (n+1)/2 \rceil}{\{\mathit{psig} \ |\ \mathit{psig} \text{ is received from $\lceil (n+1)/2 \rceil$ $\textsc{vote}$ messages for $p_i $} \}}$
        \State \hskip4em \textbf{broadcast} $\langle \textsc{vote\_proof}, p_i, vote\_proof_i \rangle$

     \smallskip
     \State \hskip2em \textbf{let} $\ell_i$ be the smallest process in $\mathit{G}_i$ that sent valid $\langle \textsc{vote\_proof} \rangle$ message to $p_i$
     \State \hskip2em \textbf{return} $\ell_i$
    
\end{algorithmic}
\end{algorithm}

We now prove the correctness of the algorithm, assuming all honest processes input the same good group $G$.

\begin{lemma}
    If $G$ is a good group, then all honest processes output the same honest process $p_j \in G$.
\end{lemma}
\begin{proof}
    Since $G$ is a good group, each honest process in $G$ will be able to collect more than $n/2$ votes while no Byzantine process can do so. Hence, each honest process in $G$ will broadcast the threshold signature formed from the votes. Therefore, all honest processes will output the smallest honest process in $G$.
\end{proof}

Next, we can observe that the algorithm terminates in $2$ rounds, where in total, $\mathcal{O}(n|G|)$ messages are sent. Furthermore, $\mathcal{O}(n|G| \kappa)$ bits are sent. Therefore,

\begin{lemma}
    \label{lemma:group_leader_auth}
    There is a group leader election algorithm for a good group $\mathit{G}$. The algorithm terminates in $2$ rounds, exchanges $\mathcal{O}(n|G|)$ messages, and exchanges $\mathcal{O}(n|G|\kappa)$ bits.
\end{lemma}

We can do a pre-processing once, and then each process can do the election locally. By doing the algorithm on $G = \Pi$, each process will receive the aggregated votes for any honest process that cannot be misclassified. Hence, by exchanging $\mathcal{O}(n^2\kappa)$ bits in $2$ rounds, all elections afterward can be done with zero communication.

\begin{lemma}
    \label{lemma:group_leader_auth_coll}
    There is a group leader election algorithm for a good group $G$. The algorithm has a pre-processing step that terminates in $2$ rounds, exchanging $\mathcal{O}(n^2)$ messages and $\mathcal{O}(n^2\kappa)$ bits. After that, all elections can be done without communication.
\end{lemma}

\subsection{Full Algorithm}\label{subsection:authenticated:ba_predictions}

Our full authenticated algorithm is presented at \Cref{algorithm:auth_ba}. First, each process runs the reduction to validated Byzantine agreement (line~\ref{line:auth_ba:certify}), where each process will obtain a pair of (value, certificate). Here, the certificate is a cryptographic object that proves the value is indeed admissible according to strong unanimity (see \Cref{section:sa_to_ev}). Thus, by employing the predicate such that $\extValid{v} = true$ if and only if $v = \langle v', cert \rangle$ where $cert$ is a certificate for $v'$, the decision from a validated Byzantine agreement would satisfy strong unanimity.

Next, we use the validated Byzantine agreement from \Cref{subsection:authenticated:vba} in a similar guess-and-double manner as our unauthenticated algorithm. The algorithm goes through phases. In each phase, the algorithm estimates $\hat{k}$, the maximum number of actual faults and misclassified processes. This estimate is then used to allocate the leader sequence used by each process. Each process will use the first $\hat{k}$ processes as the first $\hat{k}+1$ leaders in the validated Byzantine agreement, followed by $m = \Theta(\hat{k})$ processes obtained through group leader elections on $m$-grouping (see line~\ref{line:auth_ba:pre_vba_start} to line~\ref{line:auth_ba:pre_vba:end}). These ensure that when there are indeed at most $\hat{k}$ faults or misclassified processes, the validated Byzantine agreement will output a valid decision and its decision proof, due to the conditional termination of the sub-algorithm (see \Cref{lemma:vba:conditional_termination}). By doubling the estimate $\hat{k}$ in each phase, the algorithm will terminate in $\mathcal{O}(\min(\{B/n, f\}))$ rounds.

To ensure processes can halt soon enough after obtaining a valid decision, we use the fact that the validated Byzantine agreement also outputs a cryptographic object that proves the legitimacy of the decision (the decision proof). Right before deciding, a process broadcasts the decision and decision proof from the validated Byzantine agreement (line~\ref{line:auth_ba:broadcast_decision}). Therefore, all processes may decide by the end of the next phase (see line~\ref{line:auth_ba:receive_prev_decision}).

Before concluding on the algorithm description, let us remark on a small detail regarding the validated Byzantine agreement. Let us note that the agreement and external validity of the validated Byzantine agreement come from the commit certificate used therein (see \Cref{def:commit_cert}). Thus, in our algorithm, the commit certificate from the validated Byzantine agreement during a phase should be carried over to the one in the phase after. We omit this in our implementation detail for clarity. Alternatively, one can view our full algorithm as a single validated Byzantine agreement instance, where the leaders are chosen alternating between round-robin and group leader elections.

\begin{algorithm} [!ht]
\caption{Authenticated Byzantine Agreement with Classification Predictions: Pseudocode (for process $p_i$)}
\label{algorithm:auth_ba}
\begin{algorithmic} [1]
\footnotesize

\State \textbf{Input parameters:}
    \State \hskip2em $\mathsf{Value}$ $\mathit{v}_i$ 
    \State \hskip2em $\mathsf{Prediction}$ $\mathit{a}_i$

\smallskip
\State \textbf{Local variables:}
    \State \hskip2em $\mathsf{Value}$ $\mathit{certified\_value}_i \gets \bot$
    \State \hskip2em $\mathsf{Certificate}$ $\mathit{cert}_i \gets \bot$ \BlueComment{proves $\mathit{certified\_value}_i$ satisfies strong unanimity; see \Cref{section:sa_to_ev}}
    
\medskip
\State $\mathsf{AuthenticatedAgreement}(\mathit{v}_i, \mathit{a}_i)$:
    \State \hskip2em $(\mathit{certified\_value}, \mathit{cert}_i) \gets \mathsf{StrongCertification}(\mathit{v}_i)$\label{line:auth_ba:certify} \BlueComment{see \Cref{section:sa_to_ev}}
    \State \hskip2em \textbf{for} $\mathsf{Integer}$ $\phi \gets 1$ to $\lceil \log_2(t) \rceil+1$:
        \State \hskip4em \textbf{let} $m \gets \beta \cdot 2^{\phi-1}$\label{line:auth_ba:pre_vba_start} \BlueComment{$m \in \Theta(\hat{k})$ such that there is a good group}
        \State \hskip4em \textbf{let} $L_i \gets [p_1, p_2, \dots, p_{2^{\phi-1} + 1 + m}]$
        \State \hskip4em \textbf{let} $\mathit{G}_i \gets \mathsf{M\_Grouping}(\Pi, m)$ \BlueComment{see \Cref{def:m_grouping}}
        \State \hskip4em \textbf{for} $\mathsf{Integer}$ $j \gets 1$ to $m$: \BlueComment{run in parallel}
            \State \hskip6em $\mathit{L}_i[2^{\phi-1}+1+j] \gets \mathsf{GroupLeaderElection}(\mathit{G}_i[j], \mathit{a}_i)$\label{line:auth_ba:pre_vba:end}

        \smallskip
        \State \hskip4em \textbf{let} $(\mathit{decision}_i, \mathit{decision\_proof}_i) \gets \mathsf{ValidatedAgreement}(\langle \mathit{certified\_value}_i, \mathit{cert}_i\rangle, \mathit{L}_i)$ 

        \State \hskip4em \textbf{if} received valid $\langle \textsc{decision}, \mathit{d}, \mathit{d\_proof} \rangle$ from the previous phase:\label{line:auth_ba:receive_prev_decision}
            \State \hskip6em $\mathit{decision}_i \gets \mathit{d}$; $\mathit{decision\_proof}_i \gets \mathit{d\_proof}$

        \State \hskip4em \textbf{if} $\mathit{decision}_i \ne \bot$:
            \State \hskip6em \textbf{broadcast} $\langle \textsc{decision}, \mathit{decision}_i, \mathit{decision\_proof}_i \rangle$\label{line:auth_ba:broadcast_decision}
            \State \hskip6em \textbf{let} $(v, cert) \gets \mathit{decision}_i$
            \State \hskip6em \textbf{return} $v$ and \textbf{halt} 
\end{algorithmic}
\end{algorithm}

We now prove the correctness of the algorithm. We start by proving that all honest processes will decide. In fact, we will prove a stronger statement: all honest processes decide by the end of phase $\lceil \log_2(f) \rceil+1$ (recall that $f \le t$ denotes the actual number of Byzantine processes).

\begin{lemma}
    \label{lemma:auth_ba:decide_next_phase}
    Suppose an honest process decides at phase $\phi < \lceil \log_2(t) \rceil+1$. Then, each honest process decides by the end of phase $\phi+1$.
\end{lemma}
\begin{proof}
    Before deciding, the honest process broadcasts its decision and a decision proof at line~\ref{line:auth_ba:broadcast_decision}. All processes will receive them at phase $\phi+1$ at line~\ref{line:auth_ba:receive_prev_decision} and so, will decide by the end of the phase.
\end{proof}

\begin{lemma}[Termination]
    \label{lemma:auth_ba:termination}
    All honest processes decide by the end of phase $\lceil \log_2(f) \rceil+1$.
\end{lemma}
\begin{proof}    
    Suppose an honest process decides at some phase $\phi' < \lceil \log_2(f) \rceil+1$. Then, from \Cref{lemma:auth_ba:decide_next_phase}, all honest processes will decide by the end of phase $\phi'+1$. Thus, suppose that by the start of phase $\lceil \log_2(f) \rceil+1$, no honest process has decided yet. As $2^{\lceil \log_2(f) \rceil} \ge f$, there is $j$ with $1 \le j \le f+1$ where for each honest process $p_i$, $\mathit{L}_i[j] = p_\ell$ for some honest process $p_\ell$. From the conditional termination of the validated Byzantine agreement (\Cref{lemma:vba:conditional_termination}), all honest processes will obtain a decision and its decision proof from the sub-protocol and so, will decide by the end of the phase.
\end{proof}

Next, strong unanimity follows from the reduction to external validity and the employed predicate.

\begin{lemma}[Strong unanimity]
    If all honest processes propose the same value $v$, then all honest processes will decide $v$.
\end{lemma}
\begin{proof}
    From \Cref{lemma:auth_ba:termination}, all honest processes will decide. Furthermore, due to the external validity of the validated Byzantine agreement (\Cref{lemma:vba:external_validity}) and the employed $\extValid{\cdot}$ predicate, each decision is a value $v'$ that has a certificate obtained from the strong unanimity to external validity reduction (see \Cref{section:sa_to_ev}). Finally, as a certificate proves that $v'$ is admissible under strong unanimity, $v' = v$.
\end{proof}

Finally, agreement follows from the agreement of the validated Byzantine agreement.

\begin{lemma}[Agreement]
    Each honest process decides the same value $v$.
\end{lemma}
\begin{proof}
    From the agreement of the validated Byzantine agreement (\Cref{lemma:vba:agreement}), only a single value may have a decision proof. As each process decides $v$ such that there is a decision proof for $(v, \cdot)$, the lemma statement holds.
\end{proof}

Hence, the algorithm correctly solves Byzantine agreement. We now focus on the complexities.

\begin{lemma}
    \label{lemma:auth:num_phase}
    All honest processes decide within $\log_2(\min\{B/n, f\}) + \mathcal{O}(1)$ phases.
\end{lemma}
\begin{proof}
    If $B \in \Omega(n^2)$, then this follows from \Cref{lemma:auth_ba:termination}. Thus, suppose $B \in o(n^2)$. From \Cref{lemma:auth_ba:termination}, honest processes decide in at most $\lceil \log_2(f) \rceil+1$ phases. Furthermore, since $B \in o(n^2)$, there exists a phase $\phi' \le \log_2(B/n) + \mathcal{O}(1)$ where, due to \Cref{lemma:exist_1/2_good_groups} (which we use to pick the number of groups $m$), there exists a good group from the $m$-grouping done in phase $\phi'$. Hence, from the properties of the group leader election, during phase $\phi'$, there is $j$ with $2^{\phi'-1}+1 <j \le 2^{\phi'-1}+1+m$ such that for each honest process $p_i$, $L_{i}[j] = p_\ell$ for some honest process $p_\ell$. Due to the conditional termination of the validated Byzantine agreement (\Cref{lemma:vba:conditional_termination}), all honest processes will decide by the end of phase $\phi'$. Therefore, all honest processes decide in $\log_2(\min\{B/n, f\}) + \mathcal{O}(1)$ phases.
\end{proof}

Next, observe that in phase $\phi$, the size of $\mathit{L}_i$ for each honest process $p_i$ is $2^{\phi-1}+1+\Theta(2^{\phi-1})$. Hence, during the validated Byzantine agreement at phase $\phi$, each process will go through that many views. Therefore,

\begin{lemma}
    \label{lemma:auth_ba:complexities}
    Assuming each group leader election has $\mathcal{O}(1)$ rounds, all honest processes decide in $\mathcal{O}(\min\{B/n, f\})$ rounds. Furthermore, ignoring the communications done in each group leader election, honest processes exchange $\mathcal{O}(n^2)$ messages and $\mathcal{O}(n^2\kappa)$ bits.
\end{lemma}
\begin{proof}
    The size of $\mathit{L}_i$ in phase $\phi$ is $\Theta(2^{\phi})$. As honest processes decide in $\log_2(\min\{B/n, f\}) + \mathcal{O}(1)$ phases (\Cref{lemma:auth:num_phase}) and each view in the validated Byzantine agreement has $\mathcal{O}(1)$ rounds (\Cref{lemma:vba:complexities}), in total, $\mathcal{O}(\min\{B/n, f\})$ rounds are spent for the validated Byzantine agreement. Finally, the reduction to external validity has $\mathcal{O}(1)$ rounds (\Cref{lemma:sa_to_ev}), and each phase has additional $1$ round to broadcast the decision (at line~\ref{line:auth_ba:broadcast_decision}). Thus, all honest processes decide in $\mathcal{O}(\min\{B/n, f\})$ rounds.

    Next, as each view of the validated Byzantine agreement exchanges $\mathcal{O}(n)$ messages and $\mathcal{O}(n\kappa)$ bits, honest processes exchange $\mathcal{O}(n \cdot \min\{B/n, f\})$ messages and $\mathcal{O}(n \cdot \min\{B/n, f\}\kappa)$ bits in the sub-protocol. Then, the reduction to external validity exchanges $\mathcal{O}(n^2)$ messages and $\mathcal{O}(n^2\kappa)$ bits (\Cref{lemma:sa_to_ev}), and broadcasting the decision at line~\ref{line:auth_ba:broadcast_decision} incurs $\mathcal{O}(n)$ messages and $\mathcal{O}(n\kappa)$ bits per honest process. Therefore, $\mathcal{O}(n^2)$ messages and $\mathcal{O}(n^2\kappa)$ bits are exchanged.
\end{proof}

We finally conclude with our main authenticated result.

\authquad*
\begin{proof}
    We use the group leader election from \Cref{lemma:group_leader_auth_coll}. Next, recall that we use the reduction from \Cref{section:sa_to_ev}, validated Byzantine agreement from \Cref{subsection:authenticated:vba}, and \Cref{lemma:exist_1/2_good_groups} (to choose $m$). Therefore, the algorithm tolerates up to $t < (\frac{1}{2}-\epsilon)n$ faulty processes. Finally, combined with \Cref{lemma:auth_ba:complexities}, we obtain the complexities as stated.
\end{proof}

%% file: sections/conclusion.tex
\section{Conclusion}\label{section:conclusion}

In this paper, we present several implementations for Byzantine agreement with classification predictions that achieve optimal $\mathcal{O}(\min\{B/n, f\})$ round complexity for any number of prediction errors $B$. We first give an unauthenticated algorithm with optimal resiliency $t < n/3$ with bit complexity $\tilde{\mathcal{O}}(n^2)$. Then, in the authenticated setting, via threshold signature, we give an algorithm tolerating up to $t < (\frac{1}{2}-\varepsilon)n$ faults that exchanges optimal $\mathcal{O}(n^2\kappa)$ bits. 

Several questions remain. The first one would be to know whether the $\varepsilon$-gap in the authenticated setting is needed. We remark that current approaches based on misclassified processes require this gap, so an algorithm with optimal $t < n/2$ resilience would require some new techniques. 


Lastly, it would be interesting to see whether the classification predictions can help in weaker synchrony settings, such as partial synchrony.  Many of the techniques in this paper seem transferable to partially synchronous settings, and so there is clear room for future work.

%% file: appendix/missing_proof_using_class_pred.tex
\section{Missing Proofs in \Cref{section:using_predictions}}\label{section:missing_proofs_using_pred}

In this section, we provide the proofs for \Cref{lemma:exist_1/2_good_groups} and \Cref{lemma:exist_resilient}. We first prove the existence of good groups needed by our authenticated protocol.

\existgood*
\begin{proof}
    Recall that if a group is not good, then it either has a misclassified process or all of its members are Byzantine. Here, we set $c_1 = \frac{2}{\epsilon}$ and $c_2 = \frac{\epsilon}{1-\epsilon}$. Therefore, at least
    \[
        \begin{split}
        m-k-\frac{f}{\lceil \lfloor n/m \rfloor\rceil} & > m-k-\frac{\cdot(1/2-\epsilon)n}{n/m-1} \\
         & = m - k - \frac{(1/2-\epsilon)mn}{n-m} \\
         & > m - k - \frac{(1/2-\epsilon)mn}{\frac{(1-2\epsilon)n}{(1-\epsilon)}} \\
         & > m - \frac{\epsilon}{2} m - (1/2-\epsilon/2)m \\
         & = m/2
        \end{split}
    \]
    groups are good groups.
\end{proof}

Next, we prove the grouping used in our unauthenticated algorithm. A key part of the proof relies on \emph{extractor}. In particular, the extractor we use is a bipartite multigraph with good connectivity such that the vertices in the left set (corresponding to processes) have low degree (corresponding to the number of its groups).

\existresilient*
\begin{proof}

Our $m$-grouping that guarantees many resilient groups relies on a result from Dufay, Ghinea, and Paramonov \cite{dufay2026optimalca}:

\begin{theorem}[Rephrased from Theorem 36 of \cite{dufay2026optimalca}]\label{theorem:extractor-ca}
    Let $0 < \alpha < \beta < 1, \mu > 0$ be some constants. There exists $D = poly(\frac{1}{\alpha}, \frac{1}{\beta - \alpha}, \frac{1}{\mu}) > 0$ such that for any $1 \leq m \leq n$, there exists a bipartite multigraph $G = (L \cup R, E)$ such that $|L| = n$, $|R| = m$, nodes in $L$ have degree $D$, nodes in $R$ have degree $\Theta(\frac{n}{m})$. Moreover, for any $S \subset L$ such that $|S| \leq \alpha |L|$, let $Q = \{ v \in R, \text{at least a $\beta$-fraction of $v$'s neighbors comes from $S$} \}$, then $|Q| < \mu |R|$. Furthermore, $G$ can be built deterministically in polynomial time.
\end{theorem}

We apply \Cref{theorem:extractor-ca} with $\alpha = 1/3$, $\beta = 1/2$ and $\mu = 1/12$. Let $m \leq n$, there exists $D = poly(2, \frac{1}{1/2 - 1/3}, 12) \in \mathcal{O}(1)$ such that we can find a multigraph $G = (L \cup R, E)$ satisfying the conditions in the theorem. We turn it into an $m$-grouping by mapping each process to a vertex in $L$. Then we map each vertex in $R$ to a group, where the members of the group are the processes corresponding to vertices in $L$ adjacent to it. Because each node in $L$ has degree $D$, each process is a part of at most $D$ groups. 

We define $c = 12D$. Assuming $m \geq c \cdot k$, let us show that at least $5m / 6$ of the groups are resilient. Consider $C \subseteq R$ the groups in $R$ which are not resilient. A group is not resilient if it has at least one weakly misclassified process or a Byzantine majority. Therefore, we can express $C = C_1 \cup C_2$ where $C_1$ are committees with at least one weakly misclassified process and $C_2$ are committees with a Byzantine majority. Because there are $k$ weakly misclassified processes and a process is in at most $D$ groups, we have $|C_1| \leq D \cdot k$. Because $m \geq c \cdot k = 12 D \cdot k$, we get $|C_1| \leq m / 12$. We now consider $C_2$. Let $S$ be the set of Byzantine parties, then we have $|S| = f < n/3 = |L| / 3$. We remark that $C_2$ corresponds exactly to the set of vertices in $R$ with at least a $1/2$-fraction of neighbors from $S$. So, as a property of our graph given by \Cref{theorem:extractor-ca}, $|C_2| < \mu |R| = m / 12$.

Therefore, $|C| \leq |C_1| + |C_2| < m/12 + m/12 = m/6$. Thus, at most $m/6$ groups are not resilient, so at least $5m/6$ groups are resilient.

\end{proof}

%% file: appendix/strong_to_external.tex
\section{Strong Unanimity to External Validity}\label{section:sa_to_ev}

With cryptography, there is a well-known reduction from strong unanimity to external validity \cite{civit2024dolevtight, civit2024dare}. The idea is to produce a \emph{certificate}, a cryptographic object that proves a certain value is admissible according to the strong unanimity. Then, by employing Byzantine agreement with external validity that admits value in the form of $v = \langle v', cert \rangle$, where $cert$ is a certificate proving that $v'$ is admissible according to strong unanimity, we get Byzantine agreement with strong unanimity.

The implementation we use for the reduction is presented at \Cref{algorithm:strong-cert}. Note that this tolerates up to $t < n/2$ Byzantine processes. Here, there are two types of certificates. First, a threshold signature $\mu$ such that $\tssVerify{t+1}{\mu, \langle \textsc{certify}, \mathit{v}\rangle} = \mathit{true}$ for $v \ne \bot$, that proves $v$ is admissible according to strong unanimity. The last one is a threshold signature $\mu$ such that $\tssVerify{t+1}{\mu, \langle \textsc{certify}, \bot\rangle} = \mathit{true}$, which proves that any value is admissible according to strong unanimity.

\begin{algorithm} [h]
\caption{Strong Unanimity Certification: Pseudocode (for process $p_i$)}
\label{algorithm:strong-cert}
\begin{algorithmic} [1]
\footnotesize

\State \textbf{Input parameters:}
    \State \hskip2em $\mathsf{Value}$ $\mathit{proposal}_i$

\medskip
\State $\mathsf{StrongCertification}(\mathit{proposal}_i)$:
    \State \hskip2em \textbf{broadcast} $\langle \textsc{certify}, \mathit{proposal}_i, \tssShareSign{i}{t+1}{\langle \textsc{certify}, \mathit{proposal}_i\rangle} \rangle$

    \smallskip
    \State \hskip2em \textbf{if} received $t+1$ $\langle \textsc{certify}, \mathit{v}, \mathit{psig}\rangle$ for a value $\mathit{v}$:\label{line:certify:receive_t1}
        \State \hskip4em \textbf{let} $\mathit{cert} \gets \tssCombine{t+1}{\{ \mathit{psig} \ |\ \mathit{psig} \text{ is received from $t+1$ $\textsc{proposal}$ messages} \} }$
        \State \hskip4em \textbf{broadcast} $\langle \textsc{certified}, \mathit{v}, \mathit{cert} \rangle$
    \State \hskip2em \textbf{else}:
        \State \hskip4em \textbf{broadcast} $\langle \textsc{no-common}, \tssShareSign{i}{t+1}{\langle \textsc{certify}, \bot\rangle} \rangle$

    \smallskip
    \State \hskip2em \textbf{if} received valid $\langle \textsc{certified}, \mathit{v}, \mathit{cert} \rangle$:
        \State \hskip4em \textbf{return} $(\mathit{v}, \mathit{cert})$\label{line:certify:normal_return}
    \State \hskip2em \textbf{else} \BlueComment{must receive $t+1$ $\textsc{no-common}$}
        \State \hskip4em \textbf{let} $\mathit{cert} \gets \tssCombine{t+1}{\{ \mathit{psig} \ |\ \mathit{psig} \text{ is received from $t+1$ $\textsc{no-common}$ messages} \} }$
        \State \hskip4em \textbf{return} $(\mathit{proposal}_i, \mathit{cert})$\label{line:certify:no_common_return}
\end{algorithmic}
\end{algorithm}

We first prove that if $t < n/2$, any value with a certificate must be valid according to strong unanimity, and each honest process outputs a value along with a certificate supporting it.

\begin{lemma}
    If a certificate for a value $v$ exists, then $v$ is admissible according to strong unanimity.
\end{lemma}
\begin{proof}
    If $\mu$ such that $\tssVerify{t+1}{\mu, \langle \textsc{certify}, \mathit{v}\rangle} = \mathit{true}$ for $v \ne \bot$ exists, at least one honest process must have partially signed it. As $v$ is a proposal of some honest process, $v$ is admissible. Next, we argue that if all honest processes input the same value $v$, $\mu$ such that \\ $\tssVerify{t+1}{\mu, \langle \textsc{certify}, \mathit{\bot}\rangle} = \mathit{true}$ cannot exist. This is because at least $n-t \ge t+1$ honest processes will broadcast partial signatures for $v$, and so, there will be no honest process that partially signs for $\bot$.
\end{proof}

\begin{lemma}
    Each honest process outputs $(v, cert)$, where $v$ is admissible according to strong unanimity and $cert$ is a certificate for $v$.
\end{lemma}
\begin{proof}
    Suppose that at line~\ref{line:certify:receive_t1}, an honest process received $t+1$ $\textsc{certify}$ messages for a value $v$. Then, all honest processes will receive $v$ and its certificate at line~\ref{line:certify:normal_return}. Otherwise, all honest processes will broadcast $\textsc{no-common}$ messages, and since $n-t \ge t+1$, all honest processes will obtain a certificate for any value at line~\ref{line:certify:no_common_return}. 
\end{proof}

Finally, the complexities are straightforward from the algorithm.

\begin{lemma}
    \label{lemma:sa_to_ev}
    The algorithm terminates in $2$ rounds, exchanges $\mathcal{O}(n^2)$ messages, and exchanges $\mathcal{O}(n^2\kappa)$ bits.
\end{lemma}